\documentclass{msc}

\usepackage{docmute}
\usepackage{proof}
\DeclareFixedFont{\ttb}{T1}{txtt}{bx}{n}{9} 
\DeclareFixedFont{\ttm}{T1}{txtt}{m}{n}{9}  

\usepackage{color}
\definecolor{deepblue}{rgb}{0,0,0.5}
\definecolor{deepred}{rgb}{0.6,0,0}
\definecolor{deepgreen}{rgb}{0,0.5,0}
\usepackage{listings}
\usepackage{xifthen}
\ifthenelse{\isundefined{\ismain}}{
  \newcommand{\ismain}{0}
}{}

\usepackage{enumitem}

\usepackage{tikzit}
\input{hypergraph.tikzdefs}

\tikzstyle{hedge}=[fill=white, draw=black, shape=rectangle, rounded corners=2mm, inner sep=0.2mm, outer sep=-2mm, scale=0.8, minimum height=8mm, minimum width=8mm, tikzit category=hypergraph]
\tikzstyle{hedge blue}=[hedge, fill={rgb,255: red,102; green,204; blue,255}, draw=black, shape=rectangle, tikzit category=hypergraph]
\tikzstyle{node}=[fill=black, draw=black, shape=circle, minimum size=1.5mm, inner sep=0mm, tikzit category=hypergraph]
\tikzstyle{red node}=[fill=red, draw=black, shape=circle, minimum size=1.5mm, inner sep=0mm, tikzit category=hypergraph]
\tikzstyle{node highlight}=[fill=black, draw=blue, thick, shape=circle, minimum size=1.5mm, inner sep=0mm, tikzit category=hypergraph]
\tikzstyle{red node highlight}=[fill=red, draw=blue, thick, shape=circle, minimum size=1.5mm, inner sep=0mm, tikzit category=hypergraph]
\tikzstyle{yellow hedge}=[hedge, fill=yellow, draw=black, shape=rectangle, tikzit category=hypergraph]
\tikzstyle{green hedge}=[hedge, fill=green, draw=black, shape=rectangle, tikzit category=hypergraph]
\tikzstyle{small box}=[fill=white, draw=black, shape=rectangle, minimum height=6mm, minimum width=6mm, tikzit category=string diagram]
\tikzstyle{vsmall box}=[fill=black, draw=black, shape=rectangle, minimum height=4mm, minimum width=1mm, tikzit category=string diagram, inner sep=0]
\tikzstyle{medium box}=[fill=white, draw=black, shape=rectangle, minimum height=11mm, minimum width=6mm, tikzit category=string diagram]
\tikzstyle{semilarge box}=[fill=white, draw=black, shape=rectangle, minimum height=16mm, minimum width=6mm, tikzit category=string diagram]
\tikzstyle{large box}=[fill=white, draw=black, shape=rectangle, minimum height=21mm, minimum width=6mm, tikzit category=string diagram]
\tikzstyle{black dot}=[fill=black, draw=black, shape=circle, minimum size=2mm, inner sep=0mm, tikzit category=string diagram]
\tikzstyle{white dot}=[fill=white, draw=black, shape=circle, minimum size=2mm, inner sep=0mm, tikzit category=string diagram]
\tikzstyle{red dot}=[fill=red, draw=black, shape=circle, minimum size=2mm, inner sep=0mm, tikzit category=string diagram]
\tikzstyle{wlabel}=[fill=none, draw=none, shape=rectangle, tikzit category=string diagram, font={\footnotesize}, inner sep=0pt, tikzit fill={rgb,255: red,102; green,204; blue,255}, tikzit draw={rgb,255: red,102; green,204; blue,255}, yshift=0.3mm]
\tikzstyle{BRchange}=[draw=black, shape=diamond, tikzit shape=circle, tikzit fill={rgb,255: red,96; green,0; blue,0}, diamond split part fill={black,red}, inner sep=-5mm, minimum width=2.7mm, minimum height=1.7mm]
\tikzstyle{RBchange}=[draw=black, shape=diamond, tikzit shape=circle, tikzit fill={rgb,255: red,165; green,0; blue,0}, diamond split part fill={red,black}, inner sep=0, minimum width=2.7mm, minimum height=1.7mm]
\tikzstyle{dummy}=[fill=none, draw=none, shape=circle, font={\small}, inner sep=1pt, tikzit draw=blue, tikzit fill=white]
\tikzstyle{node label}=[fill=none, draw=none, shape=rectangle, tikzit fill=cyan, tikzit draw=cyan, font={\scriptsize}, tikzit shape=circle, inner sep=0pt]
\tikzstyle{empty diag}=[fill=white, draw={rgb,255: red,165; green,165; blue,165}, shape=rectangle, minimum size=1.2 cm, dashed, thick]

\tikzstyle{dashed edge}=[-, dashed, very thick]
\tikzstyle{alt sort}=[-, dashed, dash pattern=on 2pt off 0.5pt, thick, draw=red]
\tikzstyle{diredge}=[->, >={Latex[length=1.5mm]}]
\tikzstyle{diredge highlight}=[->, >={Latex[length=1.5mm]}, draw=blue, thick]
\tikzstyle{diredge highlight alt}=[->, >={Latex[length=1.5mm]}, draw=red, thick]
\tikzstyle{boundary frame}=[-, draw={rgb,255: red,170; green,170; blue,255}, dashed, fill={rgb,255: red,238; green,238; blue,255}, thick, dash pattern=on 2pt off 0.5pt]
\tikzstyle{graph frame}=[-, draw={rgb,255: red,191; green,191; blue,191}, dashed, fill={rgb,255: red,238; green,238; blue,238}, thick, dash pattern=on 2pt off 0.5pt]
\tikzstyle{venn}=[-, draw={rgb,255: red,100; green,100; blue,100}, fill={rgb,255: red,238; green,238; blue,238}, thick, opacity=0.5]
\tikzstyle{def sort}=[-]
\tikzstyle{component}=[-, draw=red, thick]
\tikzstyle{map edge}=[{|->}, >=latex, shorten <=0.5mm, shorten >=0.5mm]
\tikzstyle{hypergraph map edge}=[{|->}, draw=red, shorten <=1mm, shorten >=1mm]
\tikzstyle{cdedge}=[->]
\tikzstyle{big cdedge}=[->, very thick, >=latex]
\tikzstyle{pointer edge}=[->, draw=gray, thick]

\usepackage{amsthm,amsmath,amssymb}
 
\newcommand{\ignora}[1]{ }

\usepackage{cleveref} 

\usepackage{microtype}
\setlist[itemize]{noitemsep, topsep=0pt}
\usepackage{xspace}
\usepackage{mathtools}
\usepackage{multicol}
\usepackage[many]{tcolorbox}
\usepackage[all,2cell]{xy}
\CompileMatrices
\UseAllTwocells
\SilentMatrices
\usepackage{graphicx}
\usepackage{float}
\usepackage{fancybox}
\setenumerate{noitemsep,topsep=0pt,parsep=0pt,partopsep=0pt}
\usepackage{comment}
\usepackage{anyfontsize}
 \usepackage{xcolor,colortbl}
\usepackage{multirow}

\newcommand{\col}{C}
\def \catC {\mathbb{C}}
\def \catD {\mathbb{D}}

\def \catA {\mathbb{A}}

\def \catC {\mathbb{C}}

\def \PROP {\mathsf{PROP}} 
\def \CPROP {\mathsf{CPROP}} 

\newcommand{\sg}{\!\lower1pt\hbox{$\includegraphics[width=8pt]{graffles/greenbullet.pdf}$}\!} 
\newcommand{\sr}{\!\lower1pt\hbox{$\includegraphics[width=8pt]{graffles/redbullet.pdf}$}\!} 
\newcommand{\sbl}{\!\lower1pt\hbox{$\includegraphics[width=8pt]{graffles/blackbullet.pdf}$}\!} 

\newcommand{\CFrob}[1]{\ensuremath{\mathbf{Frob}_{#1}}\xspace}
\newcommand{\frob}{\ensuremath{\mathbf{Frob}}\xspace}

\newcommand{\perm}[1]{\mathbf{P}_{\scriptscriptstyle #1}}

\newcommand{\old}[1]{}

\newcommand{\tr}[1]{\xrightarrow{#1}}    
\newcommand{\tl}[1]{\xleftarrow{#1}}    

\newcommand{\dlcorner}{{\ar@{}[dl]|(.8){\text{\large $\urcorner$}}}}
\newcommand{\drcorner}{{\ar@{}[dr]|(.8){\text{\large $\ulcorner$}}}}

\newcommand{\synTosem}[1]{[\! [ #1 ]\! ]}
\newcommand{\frobTosem}[1]{[ #1 ]}
\newcommand{\allTosem}[1]{\langle\! \langle #1 \rangle \! \rangle}
\newcommand{\allTosembigg}[1]{\bigg\langle\!\!\!\bigg\langle #1 \bigg\rangle\!\!\!\bigg\rangle}

\newcommand{\Cospan}[1]{\mathsf{Csp}(#1)}

\def \F {\mf{F}}
\newcommand{\mf}{\mathbf}

\newcommand\symNet{\lower3pt\hbox{$\includegraphics[width=20pt]{graffles/symmetryalt.pdf}$}}
\newcommand\Idnet{\lower3pt\hbox{$\includegraphics[width=20pt]{graffles/id.pdf}$}}

\newcommand\lccB{\lower5pt\hbox{$\includegraphics[width=25pt]{graffles/rccr.pdf}$}}
\newcommand\rccB{\lower5pt\hbox{$\includegraphics[width=25pt]{graffles/lccl.pdf}$}}
\newcommand\lccn{\lower5pt\hbox{$\includegraphics[width=20pt]{graffles/cup.pdf}$}}
\newcommand\rccn{\lower5pt\hbox{$\includegraphics[width=20pt]{graffles/cap.pdf}$}}
\def \df {\ \ensuremath{:\!\!=}\ }

\DeclareMathOperator{\id}{id}

\newcommand{\Defeq}
 {\stackrel{{def}}{=}}

\newcommand{\stran}{\raise1pt\hbox{$\centerdot$}}

\newcommand{\rring}[1]{\ensuremath{\mathbb{#1}}}

\newcommand{\N}{\rring{N}}

\newcommand{\s}[2]{\stackrel{#1}{#2}}

\newcommand{\Ra}{\Rightarrow}

\renewcommand{\emptyset}{\varnothing}

\newcommand{\ladj}[2]{\ar@/^/[#1]^-{#2} \ar@{}[#1]|-%

{\ifthenelse{\equal{#1}{r}}{\bot}{%

{\ifthenelse{\equal{#1}{rr}}{\bot}{%

{\ifthenelse{\equal{#1}{l}}{\top}{%

{\ifthenelse{\equal{#1}{u}}{\dashv}{%

{\vdash}}}}}}}}}}

\newcommand{\radj}[2]{\ar@/^/[#1]^-{#2}}

\newcommand{\radjff}[2]{\ar@{_{(}->}[#1]^{#2}}

\newcommand{\pullbacktop}[4]{%

{#1} \ar@/_/[ddr]_{#4} \ar@/^/[drr]^{#2}%

\ar@{.>}[dr]|-{#3} \\}

\let\from\colon

\newcommand{\ltsred}[1]
{ \setbox0=\hbox{$\ {}^{#1}\ $}
  \setbox1=\hbox{$\longrightarrow$}
  \loop\setbox1=\hbox{$-$\kern-0.3em\unhbox1}\ifdim\wd1<\wd0\repeat
  \hbox{$\ \ \mathop{\box1}\limits^{#1}\ \ $}
}

\newcommand{\arx}[2]{\!\xymatrix@=15pt{\ar[r]^{{#1}}_{{#2}}&}\!}

\newlength{\mylength}

{\setlength{\fboxsep}{15pt}
\setlength{\mylength}{\linewidth}%
\addtolength{\mylength}{-2\fboxsep}%
\addtolength{\mylength}{-2\fboxrule}%
\Sbox
\minipage{\mylength}%
\setlength{\abovedisplayskip}{0pt}%
\setlength{\belowdisplayskip}{0pt}%
}%
{\endminipage\endSbox
\[\fbox{\TheSbox}\]}

\newcommand{\DCospan}[2]{\mathsf{Csp}_{#1}(#2)}
\newcommand{\MDACospan}[2]{\mathsf{MACsp}_{#1}(#2)}

\newcommand{\syntax}[1]{\mathbf{S}_{#1}}

\newcommand{\FTerm}[1]{\DCospan{D}{\Hyp{{\scriptscriptstyle #1}}}}
\newcommand{\Hyp}[1]{\mathbf{Hyp}_{#1}}

\newcommand{\precBA}{\ensuremath{\preceq_{\BA{}}}\xspace}
\newcommand{\precFS}{\ensuremath{\preceq_{\FS{}}}\xspace}
\newcommand{\precU} {\ensuremath{\preceq_{U}}\xspace}
\newcommand{\precM} {\ensuremath{\preceq_{M}}\xspace}
\newcommand{\precL} {\ensuremath{\preceq_{\mathcal L}}\xspace}

\newcommand{\precmu}{\ensuremath{\preceq_{\mu}}\xspace}
\newcommand{\precdelta}{\ensuremath{\preceq_{\delta}}\xspace}

\usepackage{color}
\def\bR{\begin{color}{red}}
\def\bB{\begin{color}{blue}}
\def\bM{\begin{color}{magenta}}
\def\bC{\begin{color}{cyan}}
\def\bW{\begin{color}{white}}
\def\bBl{\begin{color}{black}}
\def\bG{\begin{color}{green}}
\def\bY{\begin{color}{yellow}}
\def\e{\end{color}\xspace}

\def \poi {\,\ensuremath{;}\,}
\def \df {\ensuremath{:=}}
\def \tns {\ensuremath{\oplus}}
\def \: {\colon}

\newcommand{\fznote}[1]{\marginpar{{\bf F.Z.} #1 }}    

\newcommand\Wmult{\itikzfig{Wmult}\xspace}
\newcommand\Wcomult{\itikzfig{Wcomult}\xspace}
\newcommand\Wunit{\itikzfig{Wunit}\xspace}
\newcommand\Wcounit{\itikzfig{Wcounit}\xspace}
\newcommand\Bmult{\itikzfig{Bmult}\xspace}
\newcommand\Bcomult{\itikzfig{Bcomult}\xspace}
\newcommand\Bunit{\itikzfig{Bunit}\xspace}
\newcommand\Bcounit{\itikzfig{Bcounit}\xspace}

\newcommand{\SynToCsp}[1]{\ensuremath{\lfloor\!\lfloor{#1}\rfloor\!\rfloor}}    

\newcommand{\rrule}[2]{\ensuremath{\left\langle #1,#2 \right\rangle}}

\newcommand{\out}[1]{\mathsf{out}(#1)}
\newcommand{\inp}[1]{\mathsf{in}(#1)}

\newcommand{\node}{\lower0pt\hbox{$\includegraphics[width=6pt]{graffles/node.pdf}$}}
\newcommand{\hyperedge}{\lower2pt\hbox{$\includegraphics[width=25pt]{graffles/hyperedge.pdf}$}}
\newcommand{\ZeronetT}{\lower4pt\hbox{$\includegraphics[width=14pt]{graffles/idzerocircuit.pdf}$}}

\newcommand\idncircuit{\lower4pt\hbox{$\includegraphics[width=18pt]{graffles/idncircuit.pdf}$}}
\def \dfop {\ \ensuremath{=:}\ }

\usepackage{stmaryrd}

\newtheorem{example}[therm]{Example} 

\begin{document}

\title{String Diagram Rewrite Theory II:\\ Rewriting with Symmetric Monoidal Structure}

\begin{authgrp}
\author{Filippo Bonchi}
\affiliation{University of Pisa
\email{filippo.bonchi@unipi.it}
}

\author{Fabio Gadducci}
\affiliation{University of Pisa
\email{fabio.gadducci@unipi.it}
}

\author{Aleks Kissinger}
\affiliation{University of Oxford
\email{aleks.kissinger@cs.ox.ac.uk}
}

\author{Pawel Sobocinski}
\affiliation{Tallinn University of Technology
\email{sobocinski@gmail.com}
}

\author{ Fabio Zanasi}
\affiliation{University College London
\email{f.zanasi@cs.ucl.ac.uk}
}
\end{authgrp}

\jnlPage{1}{00}
\jnlDoiYr{2020}
\doival{10.1017/xxxxx}

\begin{abstract}
Symmetric monoidal theories (SMTs) generalise algebraic theories in a way that make them suitable to express resource-sensitive systems, in which variables
 cannot be copied or discarded at will.
%
In SMTs, traditional tree-like terms are replaced by \emph{string diagrams}, topological entities that can be intuitively thought of as
diagrams of wires and boxes. Recently, string diagrams have become increasingly popular as a graphical syntax to reason about computational models across diverse fields, including programming language semantics, circuit theory, quantum mechanics, linguistics, and control theory. In applications, it is often convenient to implement the equations appearing in SMTs as \emph{rewriting rules}. This poses the challenge of extending the traditional theory of term rewriting,
which has been developed for algebraic theories, to string diagrams.
In this paper, we develop a mathematical theory of string diagram rewriting for SMTs. Our approach exploits the correspondence between string diagram rewriting and double pushout (DPO) rewriting of certain graphs, introduced in the first paper of this series. Such a correspondence is only sound when the SMT includes a \emph{Frobenius algebra} structure. In the present work, we show how an analogous correspondence may be established for arbitrary SMTs, once an appropriate notion of DPO rewriting (which we call \emph{convex}) is identified.
As proof of concept, we use our approach to show termination of two SMTs of interest: Frobenius semi-algebras and bialgebras. \end{abstract}


\maketitle

\section{Introduction}

The study of algebraic theories and their role in modelling computing systems~\cite{hyland2007category,DBLP:journals/entcs/BehrischKP12} is a
recurring theme of John Power's research, and the subject of some of his most influential contributions.
In a series of articles~\cite{power1999enriched,DBLP:conf/calco/Power05,DBLP:journals/jfp/LackP09,DBLP:journals/lmcs/GarnerP18}, he and his coauthors developed an enriched category theoretic generalisation of Lawvere theories and explored their applications, particularly in the study of computational effects of programming languages. Whereas monads provide a powerful theory for principled and compositional definitions of \emph{denotational} semantics, as pioneered by Moggi~\cite{moggi1991notions}, algebraic theories are particularly useful~\cite{DBLP:conf/fossacs/Power04,DBLP:conf/mpc/Power06,DBLP:journals/entcs/Power06a} in the development of formal and principled approaches to \emph{operational} semantics, as shown in a series of articles as part of a long-running and productive collaboration with Gordon Plotkin~\cite{plotkin2001semantics,DBLP:conf/fossacs/PlotkinP01,DBLP:conf/fossacs/PlotkinP02,DBLP:conf/ifipTCS/HylandPP02,DBLP:journals/acs/PlotkinP03,DBLP:journals/entcs/PlotkinP04}.

There have been several efforts to generalise the notion of algebraic theory in
 general, and that of Lawvere theory in particular.
Especially after the work of Lack~\cite{Lack2004a}, the theory of
 PROPs~\cite{MacLane1965} ---a particularly simple family of symmetric strict monoidal
 categories--- has been advanced as a categorical tool for the study of algebraic theories,
 and PROPs have
 been applied in several parts of computer
science~\cite{CoeckeDuncanZX2011,BaezErbele-CategoriesInControl,SadrzadehCC14,ZanasiThesis,GhicaJL17,CoeckeKissinger_book,BonchiSZ17,BonchiHPS17,BPSZ-lics19,JacobsKZ19}.
%
The notion of algebraic theory here is that of symmetric monoidal theory, with the essential difference being that the underlying assumption of Cartesianity
is discarded.
Indeed, PROPs generalise Lawvere theories, since the latter are nothing but Cartesian PROPs. The correspondence is well-behaved enough to extend to
\emph{presentations} of theories: indeed, it has been since long understood~\cite{Fox76} that any presentation of a Lawvere theory can be seen as a presentation
of a symmetric monoidal theory~\cite{Bonchi2018}.
PROPs are more general and can be used to capture partial and relational theories~\cite{GadducciCorradini-functorial,Liberti2021,Bonchi2017c,Zanasi16}. Overall, it appears that symmetric monoidal structure is the axiomatic baseline for many pertinent examples.

One of the driving motivations for the development of rewriting theory has been the desire to implement aspects of algebraic theories. For example, the word problem for a (presentation of an) algebraic theory is decidable if one can orient the equations $l=r$, obtaining rewriting rules $l\rightarrow r$, and prove confluence and termination of the resulting rewriting system. In this way, one obtains normal forms. To decide whether two terms are judged equal by such an algebraic theory, it suffices to rewrite both until no more redexes are found, then compare the results: They are equal precisely when they rewrite to the same normal form. However, classical rewriting theory has been developed for ordinary terms, which are intimately connected with classical algebraic theories.

\medskip

The big question driving the theoretical contributions of this paper is ``how does one implement algebraic theories captured by PROPs?''.
If we take rewriting as an answer, then the rewriting has to be done up-to the axioms of symmetric strict monoidal categories.

\medskip
Traditional terms enjoy a pleasantly simple structure: their syntactic decomposition may be represented as trees. Analogously, the structure of terms of symmetric monoidal categories
may be represented as a particular family of string diagrams. However, there is an underlying problem: trees are combinatorial objects that
are elements of a classical, inductive data type, with well-understood and efficient algorithms that are exploited for rewriting. On the other hand, string
diagrams have traditionally been considered as topological entities~\cite{Joyal1991}. Our first task is therefore to understand string diagrams as combinatorial
objects. In the prequel~\cite{BGKSZ-partone} to this paper, we showed a close connection between string diagrams over a signature $\Sigma$ and
the category of discrete cospans of hypergraphs with $\Sigma$-typed edges. Nevertheless, the correspondence is not an isomorphism: for isomorphic
cospans to be equated as string diagrams, they must be considered up-to an underlying special Frobenius structure. While examples of such theories abound,
here we consider mere symmetric monoidal structure. In Theorem~\ref{thm:charactImage}, we characterise those cospans that arise via this correspondence,
and in Proposition~\ref{th:characterisation-col} we extend this characterisation to the multi-coloured case.
The cospans of interest are those whose underlying hypergraph is acyclic and satisfies an additional technical condition that we refer to as monogamy.
Checking both is algorithmically simple enough.

Having identified a satisfactory combinatorial representation leads us to the actual mechanism of rewriting, which is an adaptation of the double pushout (DPO)
approach~\cite{HandbookDPO}. The first modification of classical DPO is forced on us by the fact that we are rewriting cospans of hypergraphs,
hence the rewriting has to be done in a way that respects the interfaces. Similar approaches have been considered previously in the literature on graph rewriting,
though, the most notable example being~\cite{Ehrig2004}. Secondly, and more seriously, the general mechanism of DPO rewriting is not sound for mere symmetric
monoidal structure. The reason for this has already been highlighted in the previous paragraph: the correspondence between string diagrams and cospans of
hypergraphs works when string diagrams are considered modulo the axioms of symmetric monoidal categories as well as those of the Frobenius structure.
Given that we do not want to assume the presence of Frobenius, we must suitably restrict the DPO mechanism.
We introduce two technical modifications: first, legal pushout complements are restricted to a variant we call boundary complements in order to preserve
monogamy and  acyclicity of the resulting cospan, and second, matches have to be restricted to convex matches, which have a topologically intuitive
explanation. We call the resulting variant convex DPO rewriting and show that it is a sound and complete mechanism for rewriting modulo symmetric
monoidal structure in Theorem~\ref{th:adequacyRigidSMT}.

To illustrate the framework, we study two examples of symmetric monoidal theories: Frobenius semi-algebras and bialgebras. For both theories, we demonstrate straightforward proofs of termination making explicit use of the graph-theoretic structure. We furthermore show that the theory of Frobenius semi-algebras is not confluent using a surprising property of convex rewriting: namely that non-overlapping rewriting rule applications can interfere with each other. Our counter-example to confluence is adapted from an example due to Power~\cite{power1991npasting}, which was originally given in the context of pasting diagrams for 3-categories. Incidentally, Power's example also occurred in a festschrift, in honour of Max Kelly's 60th birthday in 1991.


\medskip
This paper is based on previous conference articles~\cite{BonchiGKSZ_lics16,BGKSZ-esop17,BonchiGKSZ18}, and it is the second in the series,
following~\cite{BGKSZ-partone}, that collects these results in a coherent and comprehensive narrative. The formulation of our characterisation for coloured props,
as well as the case studies of Frobenius semi-algebras (Section~\ref{semiFrob}) and bialgebras (Section~\ref{case:bialg}), are novel with respect to the conference papers.

\paragraph*{\emph{Structure of the paper.}} Although the material in this paper is a prosecution of~\cite{BGKSZ-partone},
we have tried to make the presentation self-contained. We give the background material in \S\ref{sec:background}.
In \S\ref{sec:combinatorial} we give the characterisation of string diagrams for PROPS as discrete cospans of hypergraphs that are acyclic and monogamous.
In \S\ref{sec:rewriting} we develop convex DPO rewriting, the mechanism that correctly implements rewriting modulo symmetric monoidal structure.
Finally, in \S\ref{sec:casestudy} we consider two case studies: Frobenius semi-algebras and bialgebras.

\section{Background}
\label{sec:background}

\subsection{Symmetric Monoidal Theories and PROPs}

In this section we recall some basic notions and fix notation. We confine ourselves to the definitions that are strictly necessary for our developments, and refer the reader  to the first part of this exposition~\cite{BGKSZ-partone} for a gentler introduction to the same notions.

\medskip

A \emph{symmetric monoidal theory} (SMT) is a pair $(\Sigma, \mathcal E)$. Here, $\Sigma$
is a \emph{monoidal signature}, consisting of operations $o \: n \to m$ of a fixed \emph{arity} $n$ and \textit{coarity} $m$, for $n,m \in \N$. The second element $\mathcal E$ is a set of \emph{equations}, namely pairs of \emph{$\Sigma$-terms} $l,r \colon v \to w$ with the same arity and
coarity. Recall that $\Sigma$-terms are constructed by combining the operations in $\Sigma$, \emph{identities} $\id_n \colon n\to n$ and \emph{symmetries} $\sigma_{m,n} \colon m+n\to n+m$ for each $m,n\in \mathbb N$, by sequential ($;$) and parallel ($\tns$) composition.

SMTs have a categorical rendition as PROPs~\cite{MacLane1965} (monoidal \textbf{pro}duct and \textbf{p}ermutation categories).

\begin{definition}[PROP]
A PROP is a symmetric strict monoidal category with objects the natural numbers, where the product on objects, denoted $\tns$, is addition. Morphisms are identity-on-objects symmetric strict monoidal functors. PROPs and their morphisms form the category $\PROP$.
\end{definition}

\begin{figure}[t]
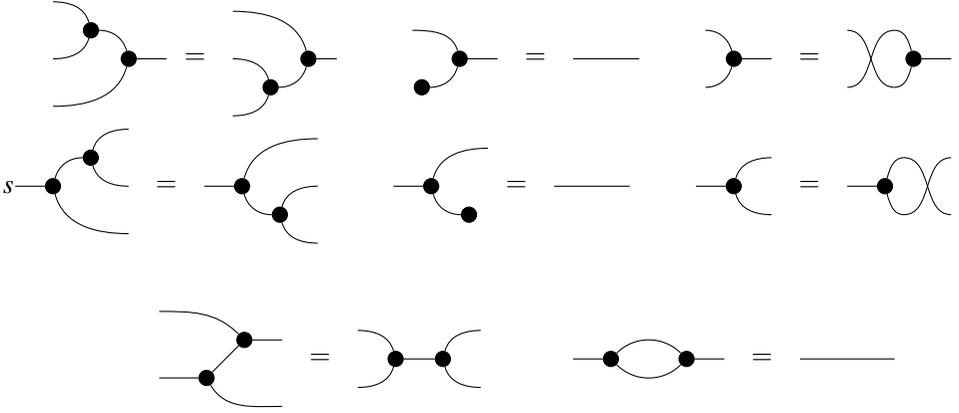

\[
\begin{array}{rcl}
(s \poi t) \poi u = s \poi (t \poi u) & \qquad \qquad &
\id_n \poi s = s = s\poi \id_m\\
(s \tns t) \tns u = s \tns (t \tns u) &  &
\id_0 \tns\ s = s  = s \tns \id_0 \\
(s \poi u) \tns (t \poi v) = (s \tns t) \poi (u \tns v) &  &
\id_m \tns \id_n = \id_{m+n} \\
(\sigma_{m,n} \tns \id_p) \poi (\id_n \tns\ \sigma_{m,p}) = \sigma_{m,n+p} &  &
\sigma_{m,n}\poi\sigma_{n,m}=\id_{m+n} \\
\end{array}
\]\vspace{-.4cm}
\[
\begin{array}{c}
(s\tns \id_m) \poi \sigma_{m,n} = \sigma_{m,p} \poi (\id_m \tns\ s) 
\end{array}
\]\vspace{1cm}
\caption{Laws of symmetric monoidal categories instantiated to a PROP ($\catC$, $\tns$, $0$), with $m,n,p \in \mathbb N$ objects of $\catC$ and
 $s,t,u,v$  morphisms of $\catC$ of the appropriate (co)arity. The laws express associativity of $\poi$ and $\tns$, and how they interact with each other and with the identities. Also, they express that symmetries are natural and involutive.}\label{fig:axSMC}\end{figure}

An SMT $(\Sigma, \mathcal E)$ gives raise to a PROP $\syntax{\Sigma, \mathcal E}$ by letting the arrows $u\to v$ of $\syntax{\Sigma, \mathcal E}$ be the $\Sigma$-terms $u\to v$ modulo the laws of symmetric monoidal categories (Figure~\ref{fig:axSMC}) and the smallest congruence containing the equations $t=t'$ for any $(t,t') \in \mathcal E$. We are going to represent these arrows by using the graphical language of \emph{string diagrams}~\cite{Selinger2009}. When $\mathcal E$ is empty, we shall use notation $\syntax{\Sigma}$ for the PROP presented by $(\Sigma, \mathcal E)$.

The SMT of \emph{special commutative Frobenius algebras} (which we shall usually refer to simply as Frobenius algebras, for brevity) plays a special role in our developments.

\begin{example}[Frobenius Algebras]
\label{ex:frob}
Consider the SMT $(\Sigma_{\textbf{\frob}}, {\mathcal E}_{\textbf{frob}})$, where
\[ \Sigma_{\frob} := \left\{\, \Bmult \colon 2 \to 1,\, \Bunit \colon 0 \to 1 ,\, \Bcomult \colon 1\to 2,\, \Bcounit \colon 1\to 0 \, \right\} \]
and ${\mathcal E}_{\textbf{frob}}$ is the set consisting of the following three equations

\begin{equation*}
\begin{split}
\input{monoid-laws.tikz} \\s
\input{comonoid-laws.tikz}\\[2em]
\input{frobenius-law.tikz} \qquad
\input{special-law.tikz}
\end{split}
\end{equation*}

We use $\frob$ to abbreviate the PROP $\syntax{(\Sigma_{\frob}, {\mathcal E}_{\textbf{frob}})}$ presented by $(\Sigma_{\frob}, {\mathcal E}_{\frob})$.

\end{example}

As for regular algebraic theories, one may consider multi-coloured versions of SMTs and PROPs. What changes is the notion of signature, which is now a pair $(\col, \Sigma)$ consisting of a set  $\col$ of \emph{colours} 
and a set $\Sigma$ of operations $o \colon w \to v$, with $w, v \in \col^{\star}$ words over $\col$.

\begin{definition}[Coloured Prop]
Given a finite set $\col$ of colours, a \emph{$\col$-coloured PROP} $\catA$ is a symmetric strict monoidal category where the set of objects $\col^{\star}$ is the set of the finite words over $\col$ and the monoidal product on objects is word concatenation. A morphism from a $\col$-coloured PROP $\catA$ to a $\col'$-coloured PROP $\catA'$ is a symmetric strict monoidal functor $H \: \catA\rightarrow \catA'$ acting on objects as a monoid homomorphism generated by a function $\col \to \col'$. Coloured PROPs and their morphisms form the category $\CPROP$.
\end{definition}

As expected, $\PROP$ is the full sub-category of $\CPROP$ given by restricting to $\{c\}$-coloured PROPs, for a fixed colour $c$.

\begin{example}\label{ex:colouredfrob}
	For later use, we recall the multi-colored analogous of Example~\ref{ex:frob}, which is the theory of Frobenius algebras over a set $\col$ of colours. Its monoidal signature  includes operations
	\[  \Bmult \colon cc \to c\qquad \Bunit \colon \epsilon \to c \qquad \Bcomult \colon c\to cc \qquad \Bcounit \colon c\to \epsilon  \]
	 (where $\epsilon \in \col^{\star}$ denotes the empty word) and equations as in Example~\ref{ex:frob} for each colour $c \in \col$. We write $\frob_{\col}$ for the $\col$-coloured PROP presented by this SMT.
\end{example}

\begin{remark}\label{rmk:pushoutCprops} 
As observed in \cite{BGKSZ-partone}, coproducts in $\CPROP$ work a bit differently than in $\PROP$. Intuitively, a coproduct $\catC + \catC'$ in $\PROP$ will identify the common core of the two PROPs, i.e. the set of objects $\N$ and the symmetrical monoidal structure on these objects. Instead, a coproduct in $\CPROP$ will not make such identification, as the involved PROPs, say $\catD$ and $\catD'$, may be based on different sets of colours. However, if $\catD$ and $\catD'$ happen to be coloured from the same set $\col$, we may still identify their common structure. Formally, this takes the form of a pushout, which we write $\catD +_{\scriptscriptstyle \col} \catD'$. Such pushout is obtained from the span of the inclusion morphisms $\catD \tl{} \perm{\col} \tr{} \catD'$, where $\perm{\col}$ is the $\col$-coloured PROP presented by the theory with an empty signature $(\col, \emptyset)$ and no equations. One may think of $\perm{\col}$ as having arrows $w \to v$ the permutations of $w$ into $v$ (thus arrows exist only when the word $v$ is an anagram of the word $w$).
\end{remark}

\subsection{Syntactic rewriting for PROPs}


\begin{definition} \label{defn:rewprop}
  A \textit{rewriting rule} in a PROP $\catA$ consists of a pair of arrows $l, r \: i \to j$ in $\catA$ with the same arities and coarities, which we write as $\rrule{l}{r}$. Given $a,b \: m \to n$ in $\catA$, $a$ rewrites into $b$ via $\RS$, written $a \Rew{\rrule{l}{r}} b$, if they are decomposable as follows
\begin{equation}
\label{eq:rewpropmatch}
\input{lhs-ctx.tikz} \qquad\qquad
\input{rhs-ctx.tikz}
\end{equation}
In this situation, we say that $a$ contains a \emph{redex} for $\rrule{l}{r}$. A \emph{rewriting system} $\RS$ is a set of rewriting rules, where we write $a \Rew{\RS} b$ to mean there exists $\rrule{l}{r} \in \RS$ such that $a \Rew{\rrule{l}{r}} b$.
\end{definition}

The equations $\mathcal{E}$ associated with an SMT $(\Sigma,\mathcal{E})$
can be oriented as rewriting rules. They give rise to a rewriting system, in the PROP
$\syntax{\Sigma}$ presented by $(\Sigma,\varnothing)$. Note that the
decompositions~\eqref{eq:rewpropmatch} are equalities modulo the laws of SMCs
(Figure~\ref{fig:axSMC}). Thus rewriting in a PROP always happens modulo these laws.

\begin{notation}
  Note that we write generic pairs and tuples using parentheses and reserve the notation $\rrule{l}{r}$ specifically for the case when the pair $(l,r)$ forms a rewriting rule.
\end{notation}

\subsection{Hypergraphs with Interfaces}

String diagrams in PROPs are interpreted as hypergraphs with interfaces, which we recall below.

\begin{definition}[Hypergraphs]
\label{defn:hyp}
A \emph{hypergraph} $G$ consists of a set $G_\star$ of \emph{nodes} and, for each $k,l \in \N$, a (possibly empty) set of \emph{hyperedges} $G_{k,l}$ with $k$ (ordered) sources and $l$ (ordered) targets of elements in $G_\star$, while a hypergraph morphisms $f: G \rightarrow H$  is a family of functions $\{f_\star, f_n \mid n \in \N\}$
satisfying the expected constraints.

We denote by $\Hyp{}$ the category of (finite) hypergraphs and hypergraph homomorphisms.
\end{definition}

Alternatively, and the characterisation will become useful later on, $\Hyp{}$ is the functor category
${\mathbb F}^{\mathbf{I}}$, where ${\mathbf{I}}$ has as objects pairs of natural numbers $(k,l) \in \N \times \N$ together with one extra object $\star$, and, for each $k,l\in\N$, there are $k+l$ arrows from $(k,l)$ to $\star$.

Nodes will be drawn as dots and a $(k,l)$ hyperedge $h$ will be drawn as a rounded box, whose connections on the left represent the list $[s_1(h), \ldots, s_k(h)]$, ordered from top to bottom, and whose connections on the right represent the list $[t_1(h), \ldots, t_l(h)]$.

\medskip

We now introduce hypergraphs with hyperedges typed in a monoidal signature $\Sigma$.
First, define $G_{\Sigma}$ as the hypergraph with just one node and for each $k,l \in \N$ the set
of $\Sigma$-operations of arity $k$ and coarity $l$ as set of hyperedges with $k$ sources and $l$ targets.
The category $\Hyp{\Sigma}$ of \emph{$\Sigma$-labeled hypergraphs} is the category whose objects
consists of an hypergraph $G$ together with a graph homomorphism
$\lambda: G \to G_{\Sigma}$, which intuitively labels  each hyperedge with an operation in $\Sigma$, while
labelled graph homomorphisms are defined accordingly.
We call such objects $\Sigma$-hypergraphs
and we visualise them as hypergraphs whose  hyperedges $h$ are labelled by $\lambda(h)$.
Observe that this definition ensures that a $\Sigma$-operation $o \: n \to m$
labels a hyperedge only when it has $n$ (resp. $m$) ordered input (resp. output) nodes.

\begin{example}
We show our notational conventions for labelled hypergraphs with the aid of an example. The hypergraph $G$ has nodes $\{ n_1, \ldots, n_8 \}$, a $(3,3)$-hyperedge $h_1$, a $(2,1)$-hyperedge $h_2$, and a $(1,0)$-hyperedge $h_3$, and the following source and target maps
\begin{equation*}
\begin{array}{l}
s_1(h_1) := v_1 \\
s_2(h_1) := v_2 \\
s_3(h_1) := v_3 \\
\end{array}
\quad
\begin{array}{l}
t_1(h_1) := v_5 \\
t_2(h_1) := v_6 \\
t_3(h_1) := v_6
\end{array}
\qquad , \qquad
\begin{array}{l}
s_1(h_2) := v_3 \\
s_2(h_2) := v_4 \\
t_1(h_2) := v_8 \\
\end{array}
\qquad , \qquad
\begin{array}{l}
s_1(h_3) := v_6 \\
\end{array}
\end{equation*}
Also, suppose $\Sigma = \{ o_1 \colon 3 \to 3, o_2 \colon 1 \to 0, o_3 \colon 2 \to 1\}$ is a monoidal signature and $o_1$, $o_2$, $o_3$ label the hyperedges of $G$ of the matching type. Then $G$ is drawn as follows
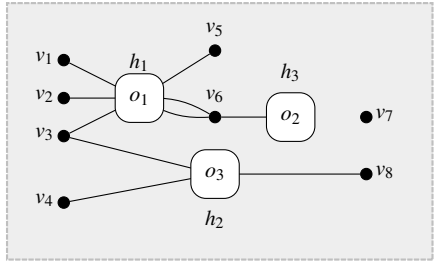
\[ \input{hypergraph-ex-labels.tikz} \]
\end{example}

\medskip
Arrows of a PROP will receive an interpretation as labeled hyergraphs \emph{with interfaces}. The notion of interface is modelled using certain \emph{cospans} in $\Hyp{\Sigma,}$.

\begin{definition}[Hypergraphs with Interfaces] \label{def:csphyp}
A \emph{cospan} from $G$ to $G'$ in $\Hyp{\Sigma}$ is a pair of arrows $G \tr{f} G'' \tl{g} G'$ in $\Hyp{\Sigma}$, where $G''$ is called the \emph{carrier} of the cospan and $G$, $G'$ are the \emph{interfaces} of $G''$. Two cospans $G \tr{f1} G_1 \tl{g1} G'$ and $G \tr{f2} G_2 \tl{g2} G'$ are \emph{isomorphic} when there is an isomorphism $\alpha \from G_1 \to G_2$ in $\Hyp{\Sigma}$ making the following diagram commute
\begin{equation*}\label{eq:preorder}
\xymatrix@R=2pt@C=10pt{
& G_1 \ar[dd]^{\alpha} & \\
G \ar@/_/[dr]_{f2} \ar@/^/[ur]^{f1} & & G' \ar@/_/[ul]_{g1} \ar@/^/[dl]^{g2} \\
& G_2 & \\
}
\end{equation*}
We define $\Cospan{\Hyp{\Sigma}}$ as the category with the same objects as $\Hyp{\Sigma}$ and arrows $G \to G'$ the cospans from $G$ to $G'$, up-to cospan isomorphism. 
Composition of $G\to  H \xleftarrow{f} G'$
and $G'\xrightarrow{g} H' \leftarrow G''$ is $G\to H+_{f,g}H' \leftarrow G''$, obtained by taking the pushout of $f$ and $g$.\footnote{Pushouts are unique only up-to iso,
 which explains why arrows of $\Cospan{\Hyp{\Sigma}}$ are defined up-to the same equivalence.}

 	We define $\DCospan{D}{\Hyp{\Sigma}}$ as the full subcategory of the category of cospans in $\Cospan{\Hyp{\Sigma}}$ with objects the discrete hypergraphs (i.e. hypergraphs with empty set of hyperedges).
\end{definition}

Notation in Definition~\ref{def:csphyp} follows the one introduced in~\cite{BGKSZ-partone}, where $\DCospan{D}{\Hyp{\Sigma}}$ is presented as an instance of a more general construction $\DCospan{H}{\catC}$, for a given functor $H$ and a category $\catC$. Without going in full details, in the case of  $\DCospan{D}{\Hyp{\Sigma}}$, the subscript $D$ is a functor with the role of selecting those cospans whose source and target (the interfaces) are \emph{discrete} hypergraphs. This means that the objects of $\DCospan{D}{\Hyp{\Sigma}}$ are natural numbers, and it is in fact a PROP.

\medskip

As with PROPs, we shall also consider the multi-coloured case of hypergraphs with interfaces. 
Given a set $\col$ of colours, a (multi-coloured) signature $(\col, \Sigma)$ can be encoded as an hypergraph $G_{\col,\Sigma}$: the set of nodes is $\col$ and each operation $o \colon u \to v$ yields an hyperedge, with ordered input nodes forming the word $u \in \col^{\star}$ and ordered output nodes forming the word $v \in \col^{\star}$. We then define the category $\Hyp{\col,\Sigma}$ of \emph{$(\col, \Sigma)$-labeled hypergraphs} as the slice category $\Hyp{} \downarrow G_{\col,\Sigma}$. Objects of $\Hyp{\col,\Sigma}$ can be visualised as hypergraphs with nodes labeled in $\col$ and hyperedges labeled in $\Sigma$, in a way that is compatible with the arity and coarity of operations in $\Sigma$.

Analogously to the one-coloured case, we can form the category $\Cospan{\Hyp{\col,\Sigma}}$ of cospans in $\Hyp{\col,\Sigma}$. We will work in  $\DCospan{D_\col}{\Hyp{\col,\Sigma}}$, the full subcategory of $\Cospan{\Hyp{\col,\Sigma}}$ with objects the discrete hypergraphs. Note that $\DCospan{D_\col}{\Hyp{\col,\Sigma}}$ is a $\col$-coloured PROP.

\subsection{Double-Pushout Rewriting of Hypergraphs}

\emph{Double-pushout (DPO) rewriting}~\cite{HandbookDPO} is an algebraic approach to rewriting that,
originally given for the category of graphs, can be defined in categories whose pushouts obey certain
well-behavedness conditions, called \emph{adhesive} categories~\cite{Lack2005}.
Now, note that $\Hyp{\Sigma}$ of Definition~\ref{defn:hyp} can be abstractly characterised as $\Hyp{} \downarrow G_{\Sigma}$, i.e., the coslice of a presheaf category: this guarantees that it is adhesive~\cite{BGKSZ-partone}, and thus we may apply DPO rewriting therein. In fact, in order to properly interpret string diagram rewriting, we will need a variation of DPO rewriting that takes into account interfaces. This variation, which we call DPO rewriting with interfaces (DPOI), has appeared in different guises in the literature, see e.g.~\cite{Ehrig2004,Gadducci07,BonchiGK09,Gadducci1998,Sassone2005a}. DPOI provides a notion of rewriting for arrows $G \tl{} J$ in $\Hyp{\Sigma}$, which we write this way to emphasise that $J$ acts as the interface of the hypergraph $G$, allowing $G$ to be ``glued'' to a context. We now recall the definition of DPOI rewriting step.

\begin{definition}[DPOI Rewriting] \label{def:dpoi}
	Given $G \leftarrow J$ and $H \leftarrow J$ in $\Hyp{\Sigma}$, \emph{$G$ rewrites into $H$ with interface $J$} --- notation $(G \tl{} J) \DPOstep{\mathcal{R}} (H \tl{} J)$ --- if there exist rule $L \tl{} K  \tr{} R$ in $\mathcal{R}$ and object $C$ and cospan of arrows $K \rightarrow C \leftarrow J$ in $\Hyp{\Sigma}$
such that the diagram below commutes and its marked squares are pushouts

\begin{equation}\label{eq:dpo2a}
\raise25pt\hbox{$
\xymatrix@R=10pt@C=20pt{
L \ar[d]_m   &  K \ar[d]
\ar@{}[dl]|(.8){\text{\large $\urcorner$}}
\ar@{}[dr]|(.8){\text{\large $\ulcorner$}}
\ar[l] \ar[r]  & R \ar[d] \\
 G &  C \ar[l] \ar[r]  & H \\
&  J \ar[u] \ar[ur]  \ar[ul]
}$}
\end{equation}
\end{definition}

Typically, DPOI rewriting takes two distinct steps: first one computes from $K \tr{} L \tr{m} G$ the object $C$ and the arrows $K \tr{} C \tr{} G$ (called a \textit{pushout complement}), then one pushes out the span $C \tl{} K \tr{} R$
 to produce the rewritten object $H$, preserving the interface $J$.

 Pushout complements always exist in $\Hyp{\Sigma}$, but they are not necessarily unique.
 They are so if the rule is \emph{left-linear}, that is, if $K \tr{} L$ is monic.
 We will come back to this point in Section~\ref{sec:rewriting}, as it plays an important role in giving a sound interpretation of string diagram rewriting as DPOI rewriting. For more details on the properties of pushout complements in $\Hyp{\Sigma}$, we refer to Part I of this work~~\cite[Section~4]{BGKSZ-partone}.

%
%
%
%

%

%
%
%
%
%
%

\section{Combinatorial Characterisation of String Diagrams}
\label{sec:combinatorial}
Let us fix a monoidal signature $\Sigma$. In \cite{BGKSZ-partone} we gave an interpretation of the arrows of the PROP
$\syntax{\Sigma}$ in terms of cospans in $\Hyp{\Sigma}$. We also saw that to make this interpretation an isomorphism, one needs to augment $\syntax{\Sigma}$ with the structure of a Frobenius algebra. Formally, there are PROP morphisms
\[
\xymatrix{
\syntax{\Sigma} \ar[r]^(0.3){\synTosem{\cdot}} & \DCospan{D}{\Hyp{\Sigma}} & \frob \ar[l]_(0.3){\frobTosem{\cdot}}
}
\]
such that their copairing $\allTosem{\cdot} \: \syntax{\Sigma} + \frob \to \DCospan{D}{\Hyp{\Sigma}}$ is an isomorphism of PROPs \cite[Theorem 3.9]{BGKSZ-partone}.
For the purpose of this paper it is convenient to recall the definition of both
morphisms: it suffices to define how they act on the generators, as for arbitrary
arrows their action is given by induction on the structure of PROPs.
The morphisms $\synTosem{\cdot} \colon \syntax{\Sigma} \to \DCospan{D}{\Hyp{\Sigma}}$ maps a generator $o \in \Sigma$ into the following cospan
\begin{equation}\label{eq:cospangenerator}
\input{hypergraph-gen.tikz}
\end{equation}
where the inputs (outputs) of the edge labelled $o$ are in bijective correspondence with the nodes of the
discrete graph on the left (on the right, respectively).

For the generators of $ \frob$, the morphism $\frobTosem{\cdot}\colon \frob \to \DCospan{D}{\Hyp{\Sigma}}$ is defined as follows
\[
\begin{array}{ccccccc}
\Bmult & \quad\mapsto\quad &\quad \input{csp-mult.tikz}
&\qquad\qquad\qquad&
\Bcomult & \quad\mapsto\quad & \quad\input{csp-comult.tikz} \\[1cm]
\Bunit & \quad\mapsto\quad & \input{csp-unit.tikz}
&\qquad&
\Bcounit & \quad\mapsto\quad & \input{csp-counit.tikz}
\end{array}
\]

%
Note that $\allTosem{\cdot} \: \syntax{\Sigma} + \frob \to \DCospan{D}{\Hyp{\Sigma}}$ is defined on the generators of $\syntax{\Sigma} + \frob$, but it respects the laws of symmetric monoidal categories (Figure~\ref{fig:axSMC}) and of Frobenius algebras (Example~\ref{ex:frob}). Indeed, a major payoff of the combinatorial interpretation is that equivalent string diagrams are all interpreted as the same hypergraph with interfaces. We refer to ~\cite{BGKSZ-partone} for more discussion on this aspect.

\medskip

In this paper we plan to exploit $\FTerm{\Sigma}$ as a combinatorial domain where to interpret rewriting in $\syntax{\Sigma}$. In the remainder of this section we thus focus on  $\synTosem{\cdot} \colon \syntax{\Sigma} \to \DCospan{D}{\Hyp{\Sigma}}$ and provide a combinatorial characterization of its image.  A preliminary series of definitions introduces the relevant hypergraph notions: \emph{monogamy} and \emph{acyclicity}.

\begin{definition}[Degree of a node]
The \textit{in-degree} of a node $v$ in hypergraph $G$ is the number of pairs $(h,i)$ where $h$ is an hyperedge with $v$ as its $i$-th target. Similarly, the \textit{out-degree} of $v$ is the number of pairs $(h,j)$ where $h$ is an hyperedge with $v$ as its $j$-th source.
\end{definition}

\begin{definition}[Monogamy] \label{def:monogamous} Given $m \tr{f} G \tl{g} n$ in $\FTerm{\Sigma}$, let $\inp{G}$ be the image of $f$ and $\out{G}$ the image of $g$. We say that the cospan is \emph{monogamous} if $f$ and $g$ are mono and for all nodes $v$ of $G$
\begin{align*}
\textrm{the in-degree of $v$ is } & \begin{cases} 0 &\mbox{if } v \in \inp{G} \\
1 & \mbox{otherwise.} \end{cases} \\
\textrm{the out-degree of $v$ is } & \begin{cases} 0 &\mbox{if } v \in \out{G} \\
1 & \mbox{otherwise} \end{cases}
\end{align*}
\end{definition}

We refer to the nodes in $\inp{G}$ and $\out{G}$ as the inputs and the outputs of $G$ and abusing notation we may
say that $G$ is monogamous.
The cospan in \eqref{eq:cospangenerator} is clearly monogamous: all the nodes on the left are inputs and they have in-degree $0$ and out-degree $1$, while all the nodes on the right are outputs and they have in-degree $1$ and out degree $0$.

\begin{example}
Four examples of cospans that are \emph{not} monogamous are displayed below. In here and the reminder of the paper, we use numeric labels when we wish to specify how the cospan legs are defined on the nodes.  
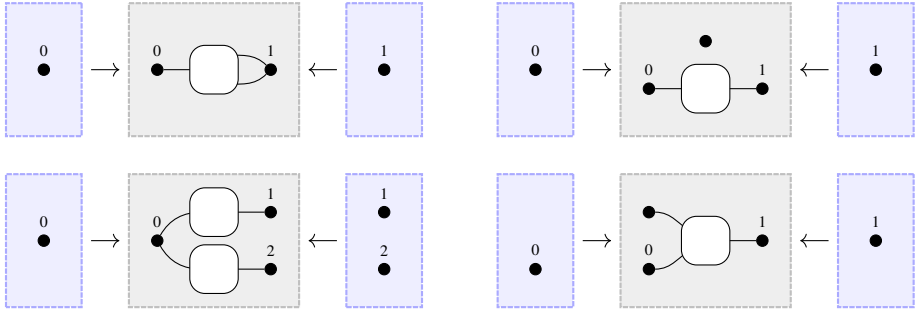
\ctikzfig{non-monog-cospans}
\end{example}

\begin{lemma}
Identities and symmetries in $ \DCospan{D}{\Hyp{\Sigma}}$ are monogamous.
\end{lemma}
\begin{proof}
The cospans identifying identities and symmetries involve discrete graphs, so the
in-degree and the out-degree
of all nodes are $0$. Moreover, all these nodes are both inputs and outputs.
\end{proof}

\begin{lemma}\label{lemma:composition}
Let $m \tr{} G \tl{} n$ and $n \tr{} H \tl{} o$ be arrows in $ \DCospan{D}{\Hyp{\Sigma}}$. If both are monogamous cospans, then  $(m \tr{} G \tl{} n) \, ;  \, (n \tr{} H \tl{} o)$ is monogamous.
\end{lemma}
\begin{proof}
Since pushouts along monos in $\Hyp{\Sigma}$ are mono, the morphisms of the cospans
 resulting from the composition $(m \tr{} G \tl{} n) \, ;  \, (n \tr{} H \tl{} o)$ are also mono.
The condition on degrees is trivially preserved since
$(m \tr{} G \tl{} n) \, ;  \, (n \tr{} H \tl{} o)$  is obtained by gluing together $G$ with $H$
along the nodes in $n$. This means that each of the nodes in $\out{G}$ is identified
with exactly one of the node in $\inp{H}$.
\end{proof}

\begin{lemma}
Let $m_1 \tr{} G_1 \tl{} n_1$ and $m_2 \tr{} G_2 \tl{} n_2$ be arrows in $\DCospan{D}{\Hyp{\Sigma}}$. If both are monogamous cospans, then  $(m_1 \tr{} G_1 \tl{} n_1) \, \tns  \, (m_2 \tr{} G_2 \tl{} n_2)$ is monogamous.
\end{lemma}
\begin{proof}
By definition $(m_1 \tr{} G_1 \tl{} n_1) \, \tns  \, (m_2 \tr{} G_2 \tl{} n_2)$ is obtained by coproduct and therefore the degree of each node is the same as in the original graphs $G_1$ and $G_2$. Moreover each node is an input iff it is an input in $G_1$ or in $G_2$ and it is an output iff it is an output in $G_1$ or $G_2$.
\end{proof}

\medskip

The notions of \emph{acyclicity} and \emph{(directed) path}
between  two nodes in a (directed) graph generalises to (directed)
hypergraphs in the obvious way.

\begin{definition}
    For a pair of hyperedges $h,h'$, we call $h$ a \textit{predecessor} of $h'$ and $h'$ a \textit{successor} of $h$ if there exists a node $v$ in the target of $h$ and in the source of $h'$.
\end{definition}

\begin{definition}[Path]\label{def:path}
    For a hypergraph $G$ and hyperedges $h,h'$, a \textit{path} $p$ from $h$ to $h'$ is a sequence of hyperedges $[h_1, \ldots, h_n]$ such that $h_1 = h$, $h_n = h'$, and for $i < n$, $h_{i+1}$ is a successor of $h_i$. We say $p$ starts at a subgraph $H$ if $H$ contains a node in the source of $h$, and terminates at a subgraph $H'$ if $H'$ contains a node in the target of $h'$. 
\end{definition}
By regarding nodes as single-node subgraphs, it clearly makes sense to talk about  paths from/to nodes as well.

\begin{definition}[Acyclicity] \label{def:directedacyclic}
    A hypergraph $G$ is \textit{acyclic} if there exists no path from a node to itself. We also call a cospan $m \tr{} G \tl{} n$ acyclic if the property holds for $G$.
\end{definition}

Like for monogamy, it is easy to see that identities and symmetries are acyclic and that the monoidal product of acyclic cospans is acyclic. Unfortunately, the composition of two acyclic cospans might be cyclic: for instance by composing the following two acyclic cospans
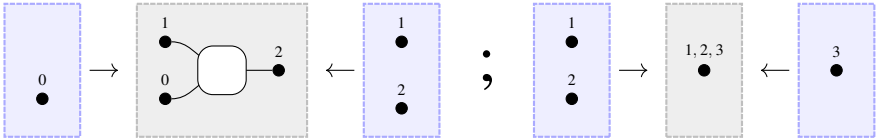
\ctikzfig{acyclic-cospan-ex}
one obtains the cyclic cospan
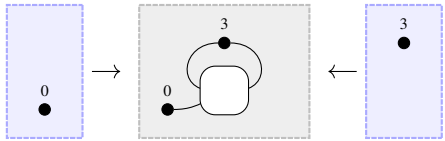
\ctikzfig{acyclic-cospan-ex2}

This issue can be avoided by additionally requiring the cospans to be monogamous.
\begin{proposition}\label{prop:subPROP}
Let $m \tr{} G \tl{} n$, $n \tr{} H \tl{} o$,  $m_1 \tr{} G_1 \tl{} n_1$ and $m_2 \tr{} G_2 \tl{} n_2$ be monogamous acyclic cospans.
\begin{enumerate}
\item Identities and symmetries in $\DCospan{D}{\Hyp{\Sigma}}$ are monogamous  acyclic.
\item $(m \tr{} G \tl{} n) \, ;  \, (n \tr{} H \tl{} o)$ is monogamous acyclic.
\item  $(m_1 \tr{} G_1 \tl{} n_1) \, \tns  \, (m_2 \tr{} G_2 \tl{} n_2)$ is monogamous  acyclic.
\end{enumerate}
\end{proposition}
\begin{proof}
The first and the third points follow from what we discussed so far.
The second point is the most interesting one. From Lemma~\ref{lemma:composition},
it follows immediately that $(m \tr{} G \tl{} n) \, ;  \, (n \tr{} H \tl{} o)$ is monogamous,
so we only need to show that this is acyclic. Since both $m \tr{} G \tl{} n$ and
$n \tr{} H \tl{} o$ are monogamous, their composition just identifies each of the nodes
in $\out{G}$ with exactly one node in $\inp{H}$.
The identification of these nodes cannot create any cycle since there is no path
in $G$ starting with one of these nodes (since their out-degree in $G$ is 0) and there
is no path in $H$ arriving in one of these nodes (since their in-degree in $H$ is 0).
\end{proof}

The above proposition informs us that monogamous  acyclic cospans form
a sub-PROP of $\DCospan{D}{\Hyp{\Sigma}}$, which we call hereafter
$\MDACospan{D}{\Hyp{\Sigma}}$. The main result of this section
(Theorem~\ref{thm:charactImage}) shows that $\MDACospan{D}{\Hyp{\Sigma}}$
is exactly the image of $\syntax{\Sigma}$ through $\synTosem{\cdot}$.
Its proof relies on an additional definition and a decomposition lemma.

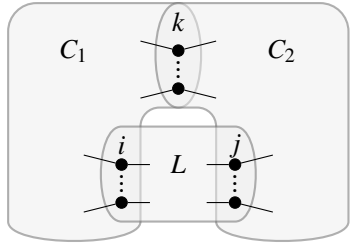
\begin{definition}[Convex sub-hypergraph] \label{def:convexsubhyp}
A sub-hypergraph $H \subseteq G$ is \textit{convex} if, for any nodes $v, v'$ in $H$ and any  path $p$ from $v$ to $v'$ in $G$,
every hyperedge in $p$ is also in $H$.
\end{definition}

\begin{example}
Consider the following hypergraph
\ctikzfig{convex-ex}
Below on the left and on the right are illustrated a convex and a non convex sub-hypergraph
\ctikzfig{convex-ex2}
\end{example}

\begin{lemma}[Decomposition]\label{lemma:convexfact}
    Let $m \rightarrow G \leftarrow n$ be a monogamous  acyclic cospan and $L$ a convex sub-hypergraph of  $G$. Then there exists $k\in\N$ and a
    unique cospan $i \rightarrow L \leftarrow j$ such that $G$ factors as
    \begin{equation}\label{eq:snd-factor}
        \input{ctx-decomposition.tikz}
    \end{equation}
    where all cospans in~\eqref{eq:snd-factor} are monogamous  acyclic.
\end{lemma}
\begin{proof}
Let $C_1$ be the smallest sub-hypergraph containing the inputs of $G$ and every
hyperedge $h$ that is not in $L$, but has a path to it. Let $C_2$ then be the smallest
sub-hypergraph containing the outputs of $G$ such that $C_1 \cup L \cup C_2 = G$.
By construction, $C_1$ and $L$ share no hyperedges. Should $C_2$ share a hyperedge
with $C_1$ or $L$, then a smaller $C_2'$ would exist such that
$C_1 \cup L \cup C_2' = G$, so $C_2$ shares no hyperedges with either
$C_1$ or $L$. Hence, the three sub-hypergraphs only overlap on nodes. Now let
    \begin{align*}
        i & := {C_1}_\ast \cap L_\ast \\
        j & := {C_2}_\ast \cap L_\ast \\
        k & := ({C_1}_\ast \cap {C_2}_\ast) \backslash L_\ast
    \end{align*}
    where $L_\ast$ are the nodes of hypergraph $L$, and the same for
    $C_1$ and $C_2$. Pictorially, these sub-hypergraphs are defined as follows
    \ctikzfig{ctx-graphs}

    Now, define the following cospans, where arrows are all inclusions
    \[ m \rightarrow C_1 \leftarrow k + i \]
    \[ i \rightarrow L \leftarrow j \]
    \[ k + j \rightarrow C_2 \leftarrow n \]
   Then \eqref{eq:snd-factor} is computed as the colimit of the following diagram
    \[ m \rightarrow C_1 \leftarrow k + i
         \rightarrow k + L \leftarrow k + j
         \rightarrow C_2 \leftarrow n \]
    The two spans identify precisely those nodes from $G$ that occur
    in more than one sub-hypergraph, so this amounts to simply taking the union
    \[ m \rightarrow C_1 \cup L \cup C_2 \leftarrow n \ =\ m \rightarrow G \leftarrow n \]
Now $C_1, C_2$, and $L$ are  acyclic because $G$ is, so it only remains
to show that each of these cospans is monogamous. For $C_1$ and $C_2$
it follows straightforwardly from the observation that, by construction, $C_1$ is
closed under predecessors and $C_2$ under successors.
%
%
The interesting case is $L$, which relies on convexity. Suppose $v$ has no
in-hyperedge in $L$. Then either $v$ is an input or there exists a hyperedge
with a path to $v$. One of these two is true precisely when $v \in i$.
Suppose $v$ has no out-hyperedge in $L$. Then, either $v$ is an ouput or
it has an out-hyperedge in $C_1$ or $C_2$. But if it has an out-hyperedge
$h$ in $C_1$, then there is a path from $v$ to another node $v'$, going
through $h$. By convexity, $h$ must then be in $L$, which is a contradiction.
Hence $v \in C_2$, which is true if and only if $v \in j$.
\end{proof}

\begin{therm}\label{thm:charactImage} A cospan $n \tr{} G \tl{} m$ is in the image of $\synTosem{\cdot}\colon \syntax{\Sigma} \to  \DCospan{D}{\Hyp{\Sigma}}$ if and only if $n \tr{} G \tl{} m$ is monogamous  acyclic. \end{therm}
\begin{proof}
The only if direction follows by induction on the arrows of $\syntax{\Sigma}$: for the base case, it is immediate to check that \eqref{eq:cospangenerator} is monogamous  acyclic, while the inductive cases follow from Proposition~\ref{prop:subPROP}.

For the converse direction, we can reason by induction on the number of hyperedges in $G$. If $G$ does not contain any, then monogamy and acyclicity imply that $n\tr{}G$ and $m\tr{}G$ are bijections, so that $n \tr{} G \tl{} m$ is in the image of an arrow only consisting of identities and symmetries.
%
For the inductive step, pick any hyperedge $e$ of $G$. Recall that $e$ has a label $l(e)\in \Sigma$ and that $l(e)$ is an arrow of $\syntax{\Sigma}$.
By monogamy and acyclicity,
$\synTosem{l(e)}$ is a convex sub-hypergraph of $G$. Hence, by Lemma~\ref{lemma:convexfact}, $n \tr{} G \tl{} m$ factors as~\eqref{eq:snd-factor}, with $L$ being $\synTosem{l( e)}$.
    The lemma guarantees that all the above cospans are monogamous  acyclic. Therefore, by the inductive hypothesis they are in the image of $\synTosem{\cdot}$, and so the same holds for $n \tr{} G \tl{} m$.
\end{proof}

The following result (Corollary 3.11 in \cite{BGKSZ-partone}) proves that
$\synTosem{\cdot}\colon \syntax{\Sigma} \to  \DCospan{D}{\Hyp{\Sigma}}$
is faithful, so an immediate consequence of the above theorem is that
$ \syntax{\Sigma}$ is isomorphic to the sub-PROP of $\DCospan{D}{\Hyp{\Sigma}}$
of monogamous  acyclic cospans.
\begin{corollary}
$ \syntax{\Sigma} \cong \MDACospan{D}{\Hyp{\Sigma}}$.
\end{corollary}
%
%
%
%
\subsection{Characterisation  for coloured PROPs}
At the beginning of this section, we recalled from  \cite{BGKSZ-partone} that  $\DCospan{D}{\Hyp{\Sigma}}$ is isomorphic to $\syntax{\Sigma} + \frob $. 
The isomorphism extends in the obvious way to the coloured case: Proposition 3.12 in \cite{BGKSZ-partone} states that $\DCospan{D_\col}{\Hyp{\col,\Sigma}}$ 
is isomorphic to $\syntax{\col,\Sigma} +_\col \CFrob{\col}$ (where $+_\col$ and $\CFrob{\col}$ are defined as in Example~\ref{ex:colouredfrob} and 
Remark~\ref{rmk:pushoutCprops}). The same holds for Theorem \ref{thm:charactImage}. The definition of $\synTosem{\cdot}_\col$ 
given in  \cite{BGKSZ-partone} is the same as the one in \eqref{eq:cospangenerator}, but with the proper interpretation of colours as labels of nodes. 
The definition of monogamous  acyclic cospans (Definitions \ref{def:monogamous} and \ref{def:directedacyclic}) does not change: the notions of degree 
and path are exactly the same in coloured and non coloured hypergraphs.
All the results proved above hold straightforwardly by following the same proofs. In particular we have the following.

\begin{therm}
\label{th:characterisation-col}
 A cospan $w \tr{} G \tl{} v$ is in the image of $\synTosem{\cdot}_\col  \: \syntax{\col, \Sigma} \to \DCospan{D_\col}{\Hyp{\col, \Sigma}}$ if and only if $w \tr{} G \tl{} v$ is monogamous  acyclic.
 \end{therm}

\section{A Sound and Complete Interpretation for String Diagram Rewriting}
\label{sec:rewriting}

In this section we develop a version of DPOI rewriting that is sound and complete for symmetric monoidal categories that \textit{do not} come with a chosen Frobenius algebra on each object. Recall that, as shown in~\cite{BGKSZ-partone}, DPOI rewriting for hypergraphs (Definition~\ref{def:dpoi}) corresponds exactly to the rewriting for a symmetric monoidal theory $\Sigma$, modulo Frobenius structure. 

Before formally stating this correspondence (Theorem~\ref{thm:frobeniusrewriting} below), we recall from~\cite{BGKSZ-partone}
 the notation $\rewiring{d}$, which refers to the ``rewiring'' of a syntactic term $d$, 
turning all of the inputs into outputs
\[
  \input{pf-syntax-d.tikz}
  \qquad \xmapsto{\rewiring{\cdot}} \qquad
  \input{pf-d-cup.tikz}
\]
Working with ``rewired'' graphs is equivalent to working with the original ones, in the sense that $d \Rew{\rrule{l}{r}} e$ if and only if $\rewiring{d} \Rew{\rrule{\rewiring{l}}{\rewiring{r}}} \rewiring{e}$.
However, since the rewired rules have only one boundary, they are readily interpreted as hypergraphs with interfaces. That is, if $d$ corresponds to a cospan $i \rightarrow G \leftarrow j$, then $\rewiring{d}$ corresponds to $0 \rightarrow G \leftarrow i + j$, or simply $G \leftarrow i + j$.

Similarly, a syntactic rewriting rule $\rrule{\rewiring{l}}{\rewiring{r}}$ corresponds to a pair of hypergraphs with the same interface, $L \leftarrow i+j$ and $R \leftarrow i+j$, i.e. a span $L \leftarrow i + j \rightarrow R$. Hence, we can extend the definition of  $\allTosem{\cdot} \: \syntax{\Sigma} + \frob \to \DCospan{D}{\Hyp{\Sigma}}$ (\emph{cf.} Section~\ref{sec:combinatorial}) to rewriting rules by letting $\allTosem{\rrule{\rewiring{l}}{\rewiring{r}}} := L \leftarrow i+j \rightarrow R$. We now have all the ingredients to recall the correspondence theorem between (syntactic) rewriting modulo Frobenius relation $\Rightarrow$ and the DPOI rewriting relation $\DPOstep{}$.

\begin{therm}[\cite{BGKSZ-partone}]\label{thm:frobeniusrewriting}
Let $\rrule{l}{r}$ be a rewriting rule on $\syntax{\Sigma}+\frob$. Then
\[
d \Rew{\rrule{l}{r}} e  \quad \text{ iff } \quad \allTosem{\rewiring{d}} \DPOstep{\allTosem{\rrule{\rewiring{l}}{\rewiring{r}}}}  \allTosem{\rewiring{e}}\text{ .}
\]
\end{therm}



In this section, we will see that the full DPOI relation $\DPOstep{}$ is not sound for rewriting in the absence of Frobenius structure. To fix this problem, 
we will put a restriction on which pushout complements (i.e. contexts) are allowed in a rewriting step. If a rule is not left-linear, there could be many different 
pushout complements for a given match, each one yielding a different result. It was shown in~\cite{BGKSZ-partone} that this makes perfect sense 
when rewriting modulo Frobenius structure. However, without this structure not all results can be interpreted in a symmetric monoidal category.

\begin{example}\label{ex:unsoundcontext}
Consider $\Sigma = \{ \alpha_1 \colon 0 \to 1, \alpha_2 \colon 1\to 0, \alpha_3 \colon 1 \to 1\}$ and the PROP rewriting system $\mathcal{R} = \left\{\ \begin{tikzpicture}
	\begin{pgfonlayer}{nodelayer}
		\node [style=none] (1) at (-1, 0) {};
		\node [style=none] (2) at (1, 0) {};
	\end{pgfonlayer}
	\begin{pgfonlayer}{edgelayer}
		\draw (1.center) to (2.center);
	\end{pgfonlayer}
\end{tikzpicture}
\  \Rew{} \ \begin{tikzpicture}
	\begin{pgfonlayer}{nodelayer}
		\node [style=small box] (0) at (0, 0) {$\alpha_3$};
		\node [style=none] (1) at (-1.25, 0) {};
		\node [style=none] (2) at (1.25, 0) {};
	\end{pgfonlayer}
	\begin{pgfonlayer}{edgelayer}
		\draw (1.center) to (2.center);
	\end{pgfonlayer}
\end{tikzpicture}
\ \right\}$ on $\syntax{\Sigma}$. Its interpretation in $\FTerm{\Sigma}$ is given by the rule
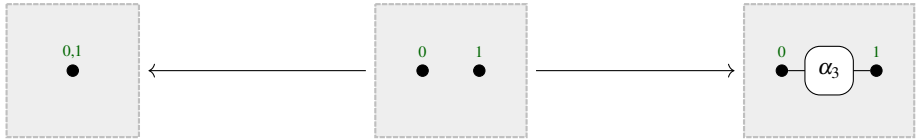
\ctikzfig{unsound-context-rule}
The rule is not left-linear and therefore pushout complements are not necessarily unique for the application of this rule. For example, the following pushout complement yields a rewritten graph that can be interpreted as an arrow in an SMC
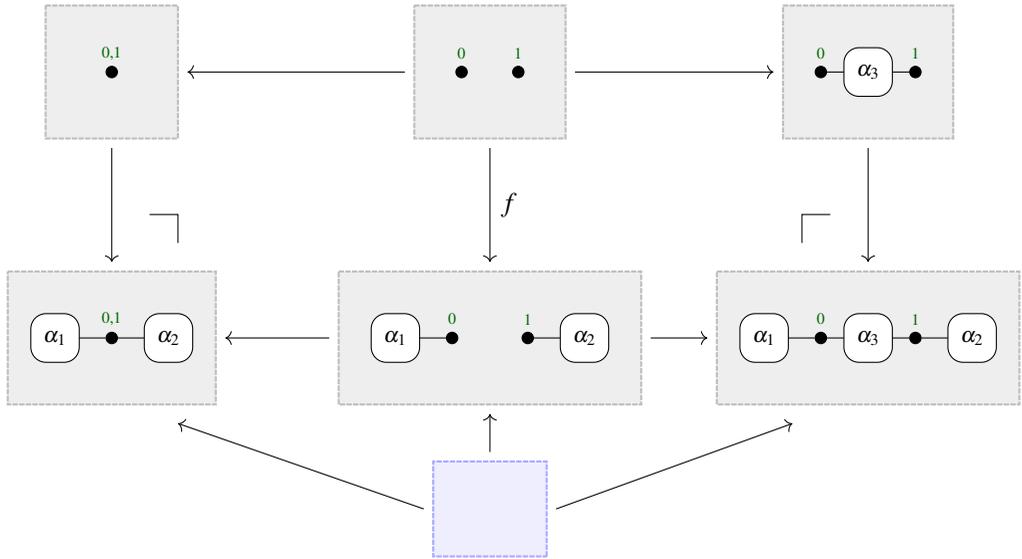
\ctikzfig{sound-context-dpo}
On the other hand, if we choose a different pushout complement, we obtain a rewritten graph that does not look like an SMC morphism
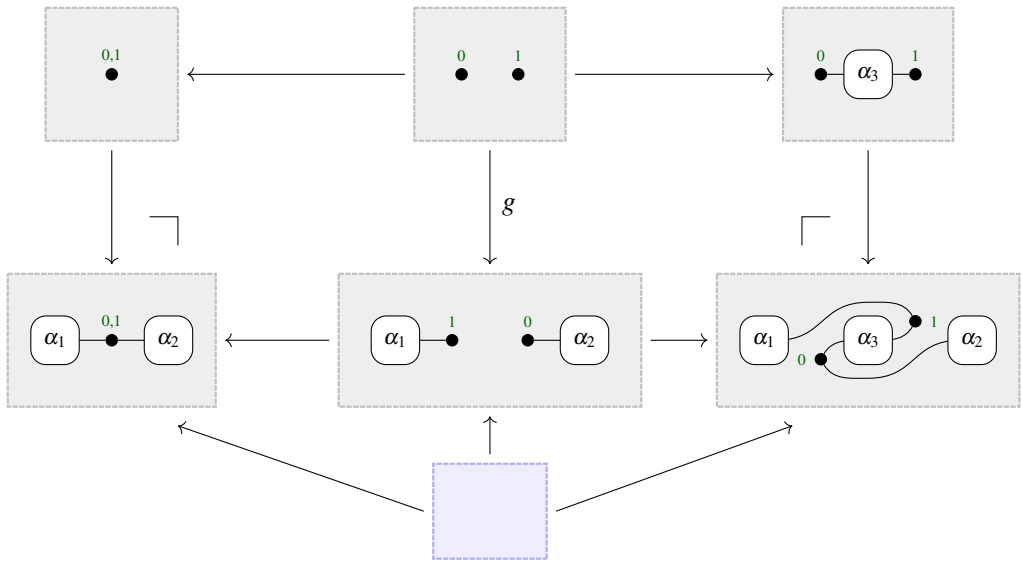
\ctikzfig{unsound-context-dpo}

The different outcome is due to the fact that $f$ maps $0$ to the leftmost and $1$ to the rightmost node, whereas $g$ swaps the assignments. Even though both rewriting steps could be mimicked at the syntactic level in $\syntax{\Sigma} + \frob$, the second hypergraph rewrite yields a hypergraph that is not in the image of any morphism of $\syntax{\Sigma}$. In particular, the rewritten graph in the second derivation is not monogamous: the outputs of $\alpha_1$ and $\alpha_3$ and the inputs of $\alpha_2$ and $\alpha_3$ have been glued together by the right pushout.
\end{example}


To rule out `bad' pushout complements, i.e. those not yielding monogamous acyclic hypergraphs after rewriting, we introduce the notion of \emph{boundary complement}, which requires that inputs only ever get glued to outputs (and vice-versa) in the two pushout squares of a DPO diagram.




\begin{definition}[Boundary complement] \label{def:boundarycomplement}
    For monogamous cospans $i \xrightarrow{a_1} L \xleftarrow{a_2} j$ and $n \xrightarrow{b_1} G \xleftarrow{b_2} m$
    and mono $f : L \to G$, a pushout complement as depicted in $(\dagger)$ below
    \[ \xymatrix@C=50pt@R=20pt{
        L \ar[d]_f \ar@{}[dr]|{(\dagger)} & {i + j} \ar[l]_{a = [a_1,a_2]}
              \ar[d]^{c = [c_1,c_2]}
        \ar@{}[dl]|(.8){\text{\large $\urcorner$}}
         \\
       G & L^\perp \ar[l]^{g} \\
       & n+m \ar[ul]^{[b_1,b_2]} \ar@{-->}[u]_{[d_1,d_2]}
       } \]
    is called a \textit{boundary complement} if $[c_1,c_2]$ is mono and there exist
    $d_1\: n\to L^\perp$ and $d_2\: m \to L^\perp$ making the above triangle commute
    and such that
    \begin{equation}\label{eq:boundary} n+j \xrightarrow{[d_1,c_2]} L^\perp \xleftarrow{[d_2,c_1]} m+i \end{equation}
    is a monogamous cospan.
\end{definition}

Intuitively, being a pushout complement, $L^\perp$ can be figured as $G$ with a `hole', filled by $L$. As $G$ and $L$ are both monogamous, 
This means that in $\syntax{\Sigma}$ we have 
\[\input{bc-context1.tikz}\]
where $g \colon n \to m$, $l \colon i \to j$, and $l^{\perp} \colon n+j \to m+i$ such that $\synTosem{g} = n \tl{} G \tr{} m$,  $\synTosem{l} = i \tl{} L \tr{} j$ and $\synTosem{l^\perp} = n+j \tl{} L^\perp \tr{} m+i$ are guaranteed to exist by monogamicity of the three cospans involved and Theorem~\ref{thm:charactImage}. Note that, in particular, the requirement that \eqref{eq:boundary} is monogamous enforces that the string diagram
\[\input{bc-context2.tikz}\]
properly lives in $\syntax{\Sigma}$ instead of $\syntax{\Sigma} + \frob$.

\medskip

The notion of boundary complement has the pleasant property of restoring uniqueness of pushout complements, even though we consider some rules which are not left-linear (namely, those with an identity morphism on the left-hand side of the syntactic rule).

\begin{proposition}\label{thm:uniquenessBoundaryCompl}
    When boundary complements in $\Hyp{\Sigma}$ exist, they are unique.
\end{proposition}
\begin{proof}
	   Since $\Hyp{\Sigma}$ is a presheaf category, a pushout of hypergraphs consists of pushouts on the underlying sets of nodes and hyperedges and an appropriate choice of source and target maps. Hence the underlying sets of $L^\perp$ must give pushout complements in the category $\F$ of finite sets
    \[
    \xymatrix{
    L_\star \ar[d]_{f_\star} &  i + j \ar[l]_{a_\star} \ar[d]^{c_\star} \ar@{}[dl]|(.8){\text{\large $\urcorner$}}  \\
    G_\star & L^\perp_\star \ar[l]^{g_\star}  \\
    }
    \qquad
    \xymatrix{
      L_{k,l} \ar[d]_{f_{k,l}} & 0 \ar[l] \ar[d] \ar@{}[dl]|(.8){\text{\large $\urcorner$}}  \\
   G_{k,l} & L^\perp_{k,l} \ar[l]^{g_{k,l}}  \\
    }
    \]
    where $G_{\star}$ and $G_{k,l}$ are the nodes and the $(k,l)$-hyperedges of $G$,
     and similarly for $L$ and $L^\perp$.
    Since $i + j$ is discrete, $G_{k,l}$ is a disjoint union, so $L^\perp_{k,l}$ must be $G_{k,l} \backslash L_{k,l}$. Since these are diagrams in $\F$, and $c,f$ are mono, we can rewrite the left pushout square as
    \[
    \xymatrix{
    l + x \ar[d]_{f_\star} &  i + j \ar[l]_{a_\star} \ar[d]^{c_\star} \ar@{}[dl]|(.8){\text{\large $\urcorner$}}  \\
    l + x + y & i + j + z \ar[l]^{g_\star}  \\
    }
    \]
    where the two downward arrows are coproduct injections and $l$ is the image of $a_\star$. One can easily verify that $l + x + z$ also gives a pushout for the given span, so we obtain a commuting isomorphism $l + x + y \cong l + x + z$, from which we conclude $y \cong z$ and that, up to isomorphism, the pushout complement on nodes must be
    \[
    \xymatrix{
   l + x \ar[d]_{f_\star}  &  i + j \ar[l]_{a_\star} \ar[d]^{c_\star} \ar@{}[dl]|(.8){\text{\large $\urcorner$}}  \\
    l + x + y  &  i + j + y \ar[l]^{g_\star} \\
    }
    \]
    from whence it follows that $g_\star = a_\star + id_n$.

\medskip
So far, we have proved that the sets $L^\perp_\star$ and $L^\perp_{k,l}$ are defined uniquely by the property of being a pushout complement. The only thing that remains is the definition of the source and target maps $s_{k,l}\: L^\perp_{k,l} \to L^\perp_\star$ and target $t_{k,l}\: L^\perp_{k,l} \to L^\perp_\star$ maps.
Since $g$ is a homomorphism, by abusing notation and denoting as $g_\star$ also its extension to sequences,
we have that for all hyperedges $h$
\[ g_\star(s_{k,l}(h)) = s_{k,l}(g_{k,l}(h)) \ \implies\ s_{k,l}(h) \in g_\star^{-1}(s_{k,l}(g_{k,l}(h))) \]
Since $g_\star$ is of the form $[(a_1)_\star, (a_2)_\star] + 1_n$, where $a_1$ and $a_2$ are mono, the inverse image $g_\star^{-1}(s_{k,l}(g_{k,l}(h)))$ contains at most two elements. In the case where it has one element, $s_{k,l}$ is uniquely fixed, so consider when it has two. It must then be the case that
\[ g_\star^{-1}(s_{k,l}(g_{k,l}(h))) = \{ v_1 \in i, v_2 \in j \} \]

But monogamy of \eqref{eq:boundary} says that the image of $i$ in $L^\perp$ cannot be the source of any hyperedge. Therefore, it must be $s_{k,l}(h) = v_2$. Similarly, if
\[ g_\star^{-1}(t_{k,l}(g_{k,l}(h))) = \{ v_1 \in i, v_2 \in j \} \]
then $t_{k,l}(h)$ must be $v_1$. Since there is at most one choice of source and target maps for $L^\perp$ making $g$ a homomorphism, $L^\perp$ must be unique.
\end{proof}


Boundary complements solve the problem highlighted in Example~\ref{ex:unsoundcontext}, in that restricting to boundary complements guarantees that the result of
a rewrite will be a monogamous acyclic hypergraph, which can be interpreted as morphisms in an SMC. However, a slightly more subtle problem remains: some graph
rewrites can be performed in such a way that the rule, target graph, and rewritten hypergraph are all monogamous acyclic, but they correspond to an equation between
morphisms that is not derivable using the SMC laws and the rules of a symmetric monoidal theory.

\begin{example}\label{ex:unsound} Consider a $\Sigma = \{ e_1 \: 1 \to 2, {e_2 \: 2 \to 1}, {e_3 \: 1 \to 1} , e_4 \: 1 \to 1\}$ and the following rewriting rule in $\syntax{\Sigma}$
  \begin{eqnarray}
    \rrule{ \ \input{unsound-lhs.tikz}\ \colon 2 \to 2 \ }{\ \ \input{unsound-rhs.tikz}\ \colon 2\to 2 } \label{eq:counterexSoundness1}
  \end{eqnarray}
  Left and right side are interpreted in $\FTerm{\Sigma}$ as cospans
  \[
    \input{unsound-lhs-cospan.tikz}
    \qquad
    \input{unsound-rhs-cospan.tikz}
  \]
  We introduce another diagram $c \: 1\to 1$ in $\syntax{\Sigma}$ and its interpretation in $\FTerm{\Sigma}$
  \[
    \input{unsound-target.tikz}
    \quad \xmapsto{\synTosem{\cdot}} \quad
    \input{unsound-target-cospan.tikz}
  \]
  Now, rule \eqref{eq:counterexSoundness1} cannot be applied to $c$, even modulo the SMC equations. However, their interpretation yields a DPO rewriting step in $\FTerm{\Sigma}$ as below
  \[
    \scalebox{0.9}{\input{unsound-dpo.tikz}}
  \]
  Observe that the leftmost pushout above \emph{is} a boundary complement: the input-output partition is correct. Still, the rewriting step cannot be mimicked at the syntactic level using rewriting modulo the SMC laws. That is because, in order to apply our rule, we need to deform the diagram such that $e_3$ occurs outside of the left-hand side. This requires moving $e_3$ either before or after the occurence of the left-hand side in the larger expression, but both of these possibilities require a feedback loop
  \[ \input{e3-before.tikz} \qquad\qquad \input{e3-after.tikz} \]
  Hence, if the category does not have at least a traced symmetric monoidal structure~\cite{Joyal_tracedcategories}, there is no way to apply the rule.
\end{example}

The source of the problem in Example~\ref{ex:unsound} is the fact that the image of the match $f$ forms a non-convex, ``U-shaped'' sub-graph of the target graph. In other words, we identify a forward-directed path of hyperedges going out of the image of $f$ and back inside again. Hyperedges in such a path (namely $e_3$ in the example) cause obstructions to rewriting in an SMC. Hence, we introduce the notion of \emph{convex} matches, which forbid forward-directed paths from outputs to inputs.


\begin{definition}[Convex match] \label{def:convexampleatching}
We call $m : L \to G$ in $\Hyp{\Sigma}$ a \textit{convex match} if it is mono and its image is convex.
\end{definition}

We saw from Lemma~\ref{lemma:convexfact} that, for any convex sub-graph  $L$ of a monogamous acyclic hypergraph $G$, $G$ can be decomposed into parts using ``$\tns$'' and ``$\poi$'', where one of those parts is $L$. This will play a crucial role in our soundness theorem.

We now combine the notions of boundary complement and of convex match to tailor a family of DPOI rewriting steps which only yield legal $\syntax{\Sigma}$-rewriting.

\begin{definition}\label{def:convex-dpoi}
Given $D \leftarrow n+m$ and $E \leftarrow n+m$ in $\Hyp{\Sigma}$, \emph{$D$ rewrites convexely into $E$ with interface $n+m$} --- notation
$(D \tl{} n+m) \rigidDPOstep{\mathcal{R}} (E \tl{} n+m)$ --- if there exist rule $L \tl{} i+j  \tr{} R$ in $\mathcal{R}$ and object $C$ and cospan arrows
$i+j \rightarrow C \leftarrow n+m$ in $\Hyp{\Sigma}$
such that the diagram below commutes and its marked squares are pushouts
%
\begin{equation}\label{eq:dpo3}
\raise25pt\hbox{$
\xymatrix@R=15pt@C=20pt{
L \ar[d]_{f}   &  i+j \ar[d]
 \ar@{}[dl]|(.8){\text{\large $\urcorner$}}
 \ar@{}[dr]|(.8){\text{\large $\ulcorner$}}
 \ar[l]_{[a_1,a_2]} \ar[r]^{[b_1,b_2]}  & R \ar[d] \\
 D &  C \ar[l] \ar[r]  & E \\
&  n+m \ar[u] \ar[ur]_{[p_1,p_2]}  \ar[ul]^{[q_1,q_2]}
}$}
\end{equation}
and the following conditions hold
 \begin{itemize}
 \item $f \: L \to D$ is a convex match;
 \item $i+j \to C \to D$ is a boundary complement in the leftmost pushout.
\end{itemize}
\end{definition}
The relation $\rigidDPOstep{\mathcal{R}}$ is contained in the DPOI rewriting relation $\DPOstep{\mathcal{R}}$ (Definition~\ref{def:dpoi}), 
the difference being that the leftmost pushout must consist of a convex match and a boundary complement.

We have now all the ingredients to prove the adequacy of convex DPO rewriting with respect to rewriting in $\syntax{\Sigma}$.

\begin{therm}\label{th:adequacyRigidSMT}Let $\mathcal{R}$ by any rewriting system on $\syntax{\Sigma}$.
Then, 
\begin{eqnarray*} d\Rightarrow_{\mathcal{R}}e & \text{ iff } & \allTosem{\rewiring{d}} \rigidDPOstep{\allTosem{\rewiring{\mathcal{R}}}} \allTosem{\rewiring{e}}.\end{eqnarray*}
\end{therm}
\begin{proof} For the only if direction, by Theorem~\ref{thm:frobeniusrewriting} $d\Rightarrow_{\mathcal{R}}e$ implies $\allTosem{\rewiring{d}} \DPOstep{\allTosem{\rewiring{\mathcal{R}}}} \allTosem{\rewiring{e}}$. One may check that the argument constructs a \emph{convex} DPO rewriting step, thus yielding the desired statement.

More in detail, our assumption gives that
\begin{eqnarray*}
  \input{pf-syntax-d.tikz} &=& \input{pf-syntax-lhs.tikz} \\
  \input{pf-syntax-e.tikz} &=& \input{pf-syntax-rhs.tikz}
\end{eqnarray*}
and hence the following equalities hold in $\syntax{\Sigma} + \frob$, where $c^{\star}_i$, $i \in \{1,2\}$, is notation for $\scalebox{0.7}{\input{cstar.tikz}}$  
\begin{eqnarray}
  \input{pf-d-cup.tikz} &=& \input{pf-syntax-lhs-cup.tikz} \label{eq:rewcompld}\\
  \input{pf-e-cup.tikz} &=& \input{pf-syntax-rhs-cup.tikz} \label{eq:rewcomple}
\end{eqnarray}
We now define
\begin{eqnarray*}
  \left( 0 \tr{} D \tl{[q_1,q_2]} n+m\right) &
  \df &
  \allTosem{\rewiring{d}} = \allTosembigg{\ \input{pf-d-cup.tikz}\ } \\
  \left( 0 \tr{} E \tl{[p_1,p_2]} n+m\right) &
  \df &
  \allTosem{\rewiring{e}} = \allTosembigg{\ \input{pf-e-cup.tikz}\ } \\
  \left( 0 \tr{} L \tl{[a_1,a_2]} i+j\right) &
  \df &
  \allTosem{\rewiring{l}} = \allTosembigg{\ \input{pf-l-cup.tikz}\ } \\
  \left( 0 \tr{} R \tl{[b_1,b_2]} i+j\right) &
  \df &
  \allTosem{\rewiring{r}} = \allTosembigg{\ \input{pf-r-cup.tikz}\ } \\
  \left( i+j \tr{} C \tl{} n+m \right) &
  \df &
  \allTosembigg{\ \ \input{pf-syntax-ctx.tikz}\ \ }. \\
\end{eqnarray*}
By these definitions and \eqref{eq:rewcompld}-\eqref{eq:rewcomple} it follows that
\begin{eqnarray*}
\left( 0 \tr{} D \tl{} n+m \right) = \left( 0 \tr{} L \tl{} i+j\right) \poi \left( i+j \tr{} C \tl{} n+m\right) \\
\left( 0 \tr{} E \tl{} n+m \right) = \left( 0 \tr{} R \tl{} i+j\right) \poi \left( i+j \tr{} C \tl{} n+m\right).
\end{eqnarray*}
Since composition of cospans is defined by pushout, we have a commutative diagram with two pushouts as in~\eqref{eq:dpo3}. It remains to check that the match $L \to D$ is convex and that $C$ is a boundary complement: these conditions can be verified by definition of the involved components. Therefore, $\allTosem{\rewiring{d}} \rigidDPOstep{\allTosem{\rewiring{\mathcal{R}}}} \allTosem{\rewiring{e}}$ by application of the rule $\rrule{\allTosem{\rewiring{l}}}{\allTosem{\rewiring{r}}}$.


We now turn to the converse direction. Let 
\begin{eqnarray*}
\allTosem{\rewiring{d}} \dfop \vcenter{\xymatrix@R=5pt@C=3pt{ & D& \\ 0 \ar[ur] && \ar[ul]_{[q_1,q_2]} n+m}} & & \allTosem{\rewiring{e}} \dfop \vcenter{\xymatrix@R=5pt@C=3pt{ & E & \\ 0 \ar[ur] && \ar[ul]_{[p_1,p_2]} n+m}}.
\end{eqnarray*}
Our assumption gives us a diagram as in~\eqref{eq:dpo3}, with application of a rule $\rrule{\allTosem{\rewiring{l}}}{\allTosem{\rewiring{r}}}$ in $\allTosem{\rewiring{\mathcal{R}}}$. We now want to show that $d\Rightarrow_{\mathcal{R}}e $ with rule $\rrule{l}{r}$, say of type $(i,j)$. Now, because $n \tr{q_1} D \tl{q_2} m = \synTosem{d}$, it is monogamous  acyclic by Theorem~\ref{thm:charactImage}. Since the match $f \: L \to D$  in~\eqref{eq:dpo3} is convex, Lemma~\ref{lemma:convexfact} yields a decomposition of $n \tr{q_1} D \tl{q_2} m$ in terms of monogamous  acyclic cospans
\begin{equation*}
\left( n \tr{} C_1 \tl{} i\!+\!k \right) \poi\!
\begin{array}{cc}
\left( k \tr{id} k \tl{id} k \right)\\
\tns \\
\left( i \tr{} L \tl{} j \right)
\end{array}
\poi
\left( j\!+\!k\tr{} C_2 \tl{} m \right).
\end{equation*}
Applying again Theorem~\ref{thm:charactImage} we obtain $c_1$, $c_2$ in $\syntax{\Sigma}$ such that
\begin{align*}
\synTosem{c_1} = n\tr{} C_1 \tl{} i+k \qquad \synTosem{c_2} = j+k\tr{} C_2 \tl{} m.
\end{align*}
By functoriality of $\synTosem{\cdot}$,
$\synTosem{d} = \synTosem{c_1 \poi (id \tns l) \poi c_2} $ and,
since $\synTosem{\cdot}$ is a faithful PROP morphism,
$d = c_1 \poi (id \tns l) \poi c_2$.
Thus we can apply the rule $\rrule{l}{r}$ on $e$, which yields
$e = c_1 \poi (id \tns r) \poi c_2$ such that $d \Ra_{\mathcal{R}} e$.
We can conclude that $\synTosem{e} = n\tr{p_1} E \tl{p_2} m$ because boundary
complements are unique (Proposition~\ref{thm:uniquenessBoundaryCompl}).
\end{proof}

Hence, we have shown soundness and adequacy of convex DPOI rewriting for symmetric monoidal theories. In other words, whenever we want to perform rewriting in a free symmetric monoidal category, we could just as well do convex DPOI.

\begin{remark}\label{rem:efficient-convex-dpoi}
  A natural question is ask is whether we can do convex DPOI rewriting efficiently. This is not obvious since computing the match in a DPOI step involves solving a subgraph isomorphism problem and, as we saw in Section 4.5 of Part 1~\cite{BGKSZ-partone}, enumerating pushout complements can require a substantial amount of computation for general. The issue with matches is not really a problem since we consider rewriting with rules whose left-hand side is of fixed constant size, which is typically much smaller than the target graph. In this regime, efficient subgraph isomorphism algorithms exist going back (at least) to Ullmann~\cite{ullmann1976algorithm} and can be easily adapted to our setting.

  In fact, we can do even better in the case of monogamous hypergraphs. One can construct a homomorphism $m : L \to G$ by traversing the nodes and hyperedges of $L$ and mapping them one-by-one. At each step in the traversal, the image of the next node (resp. hyperedge) is uniquely fixed by the image of an adjacent hyperedge (resp. node), so the match will be uniquely fixed by the image of a single node in each connected component of $L$. Hence, if $L$ is connected, we can fix a starting node $v$ in $L$ and check if for each node $v'$ in $G$ the setting $m(v) := v'$ yields a valid match in time linear in $L$. Once a match is found, we can check convexity by computing the successors of the outputs of $L$ in $G$ and checking whether any input of $L$ is contained in that set, which has worse-case complexity $O(|G_\star|)$, since we need to visit each node in $G$ at most once. That is, we can enumerate matches of $L$ and $G$ in $O(|L_\star||G_\star|^2)$ time.

  The second issue is solved for convex DPOI by requiring pushout complements to be boundary complements, which are unique, as we saw in Proposition~\ref{thm:uniquenessBoundaryCompl}. Whereas in the general case, we may have to search an exponential space of potential pushout complements, boundary complements force the fact that there is at most one solution, which can be constructed efficiently from the maps $i + j \to L \to G$. Hence convex DPOI rewriting is amenable to efficient implementations.
\end{remark}

We conclude this section by showing that, for certain well-behaved rewriting systems, convexity of matches follows automatically.

\begin{definition}\label{def:left-connected} A monogamous acyclic cospan
$n \tr{f} G \tl{g} m$
is \emph{strongly connected} if for every input $x \in f(n)$ and output $y \in g(n)$ there exists a  path from $x$ to $y$. A DPO rewriting system is \textit{left-connected}
if it is left-linear and, for every rule $L \leftarrow i + j \rightarrow R$, the induced cospans $i \rightarrow L \leftarrow j$ and $i \rightarrow R \leftarrow j$ are monogamous
acyclic and $i \to L \leftarrow j$ is strongly connected. We call a PROP rewriting system $\mathcal{R}$ on $\syntax{\Sigma}$ left-connected if for every
$\rrule{l}{r} \in \mathcal{R}$ the associated DPO rule $\allTosem{\rrule{\rewiring{l}}{\rewiring{r}}}$ is left-connected.
\end{definition}

In Definition~\ref{def:left-connected}, strong connectedness prevents non-convex
matches as in Example~\ref{ex:unsound}, whereas left-linearity guarantees
uniqueness of the pushout complements, and prevents the problem in
Example~\ref{ex:unsoundcontext}. We are then able to prove the following theorem,
for the not necessarily convex DPOI rewriting relation $\DPOstep{}$.

%
\begin{therm}\label{thm:stronglyconnected}
Let $\mathcal{R}$ be a left-connected rewriting system on $\syntax{\Sigma}$. Then
\begin{enumerate}
  \item if $d \Rew{\mathcal{R}} e$
    then
    $\allTosem{\rewiring{d}}
    \rigidDPOstep{\allTosem{\rewiring{\mathcal{R}}}}
    \allTosem{\rewiring{e}}$;
  \item if
    $\allTosem{\rewiring{d}}
    \rigidDPOstep{\allTosem{\rewiring{\mathcal{R}}}}
    \allTosem{\rewiring{e}}$
    then $d \Rew{\mathcal{R}} e$.

\end{enumerate}
\end{therm}

\begin{proof}
  (1) follows from Theorem~\ref{thm:frobeniusrewriting}. For (2), suppose $\allTosem{\rewiring{d}} \cong G \leftarrow n + m$ and the rewriting relation arose from applying a left-connected rule $L \leftarrow i + j \rightarrow R$ at match $p : L \to G$. By left-connectedness, there exists a path from every input of $L$ to every output. Hence, this will also be the case for the sub-graph $p(L)$. If there was a directed path from an output of $p(L)$ to an input, this would induce a directed cycle. But since $G$ is a monogamous acyclic hypergraph, this cannot be the case. Hence, $m(L)$ is a convex sub-graph of $G$.

  Furthermore, since the rewriting rule is left-linear, $L \leftarrow i + j$ is mono. Hence, we can compute the (unique) pushout complement by removing the hyperedges and non-interface nodes of $p(L)$ from $G$. Since $L$ and $G$ are both monogamous acyclic hypergraphs, the resulting pushout complement will always be a boundary complement. Hence the DPOI rewriting step must in fact be a convex DPOI step, so we can apply Theorem~\ref{th:adequacyRigidSMT} to complete the proof.
\end{proof}

As a consequence of this theorem, if a rewriting system is left-connected, we can forego the convexity check mentioned in Remark~\ref{rem:efficient-convex-dpoi}, so we can enumerate matches of a single rule with left-hand side $L$ in $G$ in time $O(|L_\star||G_\star|)$.

\subsection{Characterisation for coloured PROPs} It is a routine exercise to generalise the results in this section to coloured props. First, fixed a set $\col$ of colours and a monoidal signature $\Sigma$ on $\col$, \cite{BGKSZ-partone} also states a multi-coloured version of Theorem~\ref{thm:frobeniusrewriting}, proving a correspondence between rewriting in $\syntax{\col, \Sigma}$ modulo the equations of $\frob_{\col}$, and DPOI rewriting in $\Hyp{\col, \Sigma}$. One may then define convex DPOI rewriting in $\Hyp{\col, \Sigma}$, in the same way as we did for $\Hyp{\Sigma}$, and show correspondence results between this and rewriting in $\syntax{\col, \Sigma}$, analogous to Theorem~\ref{th:adequacyRigidSMT} and \ref{thm:stronglyconnected}: the colouring on nodes does not affect how these characterisations are formulated and proven.

\begin{therm}
Let $\mathcal{R}$ by any rewriting system on $\syntax{\col,\Sigma,}$.
Then 
\[ d\Rightarrow_{\mathcal{R}}e  \text{  iff  }  \allTosem{\rewiring{d}} \rigidDPOstep{\allTosem{\rewiring{\mathcal{R}}}} \allTosem{\rewiring{e}}.\]
	Furthermore, let $\mathcal{R}$ be left-connected. Then
\begin{enumerate}
  \item if $d \Rew{\mathcal{R}} e$
    then
    $\allTosem{\rewiring{d}}
    \rigidDPOstep{\allTosem{\rewiring{\mathcal{R}}}}
    \allTosem{\rewiring{e}}$;
  \item if
    $\allTosem{\rewiring{d}}
    \rigidDPOstep{\allTosem{\rewiring{\mathcal{R}}}}
    \allTosem{\rewiring{e}}$
    then $d \Rew{\mathcal{R}} e$.
\end{enumerate}
\end{therm}

\section{Case Studies}
\label{sec:casestudy}

The two most fundamental properties of interest for a rewriting system are \textit{termination} and \textit{confluence}. A rewriting relation is terminating if it admits no infinite 
sequence of rewrites, and it is confluent if any pair of hypergraphs (or terms, etc.) arising from $G$ by a sequence of rewriting steps can eventually be rewritten to the same 
hypergraph. Taken together, these properties imply the existence of unique normal forms.\footnote{For background on termination and confluence in term rewriting systems, 
see for instance~\cite{Terese03}. Termination for string diagram rewriting has been studied as an instance of higher-dimensional term rewriting, 
see~\cite{guiraud2006termination} and the discussion in Section~\ref{sec:conclusions}.}

We will now apply the framework we have developed to two specific symmetric monoidal theories: Frobenius semi-algebras and bialgebras. For both of these structures, we construct the associated DPOI rewriting system and show that it is terminating. We will also show that the first theory is not confluent, by adapting a counter-example due to Power to the setting of convex rewriting. The second theory is confluent, but we leave the proof for the sequel paper, where we develop critical pair analysis for convex rewriting~\cite{BGKSZ-partthree}.

As with term rewriting theory, an important tool for termination proofs is that of \textit{reduction orderings}. For a preorder $\preceq_a$ on hypergraphs, we can define the associated equivalence relation $\sim_a$ and the strict ordering $\prec_a$ as follows
\[
  G \sim_a H \iff (G \preceq_a H \wedge H \preceq_a G)
    \qquad\qquad
    G \prec_a H \iff (G \preceq_a H \wedge G \not\sim_a H)
\]

\begin{definition}\label{def:red-ord}
  A preorder $\preceq_a$ is called a \textit{reduction ordering} for a rewriting system $\mathcal R$ if it is well-founded (i.e.~has no infinite decreasing chains with respect to $\prec_a$) and
  \[
    G \rigidDPOstep{\mathcal R} H \
    \implies\
    H \prec_a G
  \]
  Similarly, a preorder $\preceq_a$ is called a \textit{weak reduction ordering} for $\mathcal R$ if it is well-founded and
  \[
    G \rigidDPOstep{\mathcal R} H \
    \implies\
    H \preceq_a G
  \]
\end{definition}

%
Clearly the existence of a reduction ordering forbids infinite sequences of rewrites, hence any $\mathcal R$ that admits a reduction ordering is terminating.

A common strategy in termination proofs is to define reduction orderings in pieces which are then combined lexicographically. For pre-orders $\preceq_a$ and $\preceq_b$, we define the \textit{lexicographic ordering} $\preceq_{a,b}$ as follows
\[
  H \preceq_{a,b} G \iff
    (H \prec_a G \vee (H \sim_a G \wedge H \preceq_b G))
\]
The following lemma can be shown straightforwardly from the definitions above.

\begin{lemma}\label{lem:combine-red-ord}
  For a rewriting system $\mathcal R = \mathcal R_1 \cup \mathcal R_2$, if $\preceq_a$ is a reduction ordering for $\mathcal R_1$, $\preceq_a$ is a weak reduction ordering for $\mathcal R_2$, and $\preceq_b$ is a reduction ordering for $\mathcal R_2$, then $\preceq_{a,b}$ is a reduction ordering for $\mathcal R$.
\end{lemma}

\begin{proof}
  Using the definition of a lexicographic ordering, we can see that the associated strict ordering can be expressed as follows
  \[
    H \prec_{a,b} G \iff (H \prec_a G \vee (H \sim_a G \wedge H \prec_b G))
  \]
  From this, we see that $\preceq_{a,b}$ is well-founded whenever $\preceq_a$ and $\preceq_b$ are.
  Then, if $G \rigidDPOstep{r_1} H$ for some $r_1 \in \mathcal R_1$, then $H \prec_a G$, so $H \prec_{a,b} G$. If $G \rigidDPOstep{r_2} H$ for some $r_2 \in \mathcal R_2$, then $H \preceq_a G$ and $H \prec_b G$. From $H \preceq_a G$, it is either the case that $H \prec_a G$ or $H \sim_a G$, which in either case yields $H \prec_{a,b} G$.
\end{proof}

With this bit of rewriting theory in hand, we are ready to look at our two case studies.

\subsection{Frobenius semi-algebras}
\label{semiFrob}

\textit{Frobenius semi-algebras} are Frobenius algebras lacking the unit and counit equations. That is, they are the free PROP generated by the signature
\[
  \left\{\ \mu := \Wmult, \delta := \Wcomult\ \right\}
\]
modulo the following equations
\begin{equation}\label{eq:fsa-rules}
  \input{fsa-rules.tikz}
\end{equation}
It is interesting to study such structures, because full (co)unital Frobenius algebras always induce a compact closed structure. Hence, categories that are not compact closed, such as infinite-dimensional vector spaces, do not in general have Frobenius algebras. They can nevertheless have Frobenius semi-algebras, which form the basis of interesting algebraic structures relevant to quantum theory, such as H*-algebras~\citep{abramsky2012h}.

Since Frobenius semi-algebras lack many of the equations of a Frobenius algebra, we cannot use the technique for rewriting modulo Frobenius developed in~\cite{BGKSZ-partone}. Nevertheless, we can represent this theory using a hypergraph rewriting system $\FS{}$, defined as follows
\begin{center}
  \scalebox{0.7}{\input{fsa-rules-dpo.tikz}}
\end{center}

We first give a proof of termination for this rewriting system. This would be quite involved if we wished to prove it using syntactic rewriting, modulo the equations of an SMC, but here we show it is relatively straightforward, using some graph-theoretic reduction orderings.

We first deal with (co)associativity. It should be the case that na\"ively applying rules $\FS{1}$ and $\FS{2}$ will eventually terminate with all trees of multiplications and comultiplications associated to the right (or bottom, as we are reading diagrams left-to-right). More formally, for any vertex $x$, let a $\mu$-tree with root $x$ be a maximal tree of $\mu$-hyperedges with output $x$. Similarly, a $\delta$-tree with root $x$ is a maximal tree of $\delta$-hyperedges with input $x$.


For a $\mu$-hyperedge $h$, let the $\mathcal{L}$-weight $\ell(h)$ be the size of the $\mu$-tree whose root is the first input of $h$. Similarly, for a $\delta$-hyperedge, let $\ell(h)$ be the size of the $\delta$-tree whose root is the first output of $h$. Let $\ell(h) = 0$ otherwise and
\[
  \mathcal L(G) := \sum_{h \in G_{2,1}\,\cup\, G_{1,2}} \ell(h)
\]

\begin{lemma}\label{lem:fs-L}
  The following is a reduction ordering for $\{\FS 1, \FS 2\}$
  \[
    H \preceq_{\mathcal L} G \iff
      \mathcal L(H) \leq \mathcal L(G)
  \]
\end{lemma}

\begin{proof}
  $\mathcal L$ is $\mathbb N$-valued, so $\preceq_{\mathcal L}$ is well-founded.

  Applying the rule $\FS 1$ has no effect on the $\mathcal{L}$-weight of any $\mu$-hyperedges outside of the image of the left-hand side. 
  Suppose there are $\mu$-trees of size $a, b, c$ connected to inputs $0, 1, 2$ of the left-hand side, respectively. 
  The $\mathcal{L}$-weight of the two $\mu$-hyperedges on the left-hand side are thus $a$ and $a + b + 1$, whereas on the right-hand side they 
  are $a$ and $b$. Hence $ \preceq_{\mathcal L}$ is strictly decreased by $\FS 1$. The property for $\FS 2$ follows symmetrically.
\end{proof}

The previous result accounts for associativity of $\mu$ and of $\delta$, but we should do the same with the two Frobenius equations. 
We can use the fact that each of the Frobenius rewriting
rules $\FS{3}$ and $\FS{4}$ strictly decreases $|\mathcal{D}(G)|$ where
\[
  \mathcal D(G) := \{ (h \in G_{2,1}, h' \in G_{1,2}) \ |\ \textrm{there is no path from $h$ to $h'$} \}
\]
Following Definition~\ref{def:path}, we will use the term \textit{path} in this and the next section to refer exclusively to directed paths, i.e. sequences of hyperedges $[h_1, \ldots, h_n]$ such that $h_{i+1}$ is a successor of $h_i$. Note that, since $h \in G_{2,1}$, it must be a $\mu$-hyperedge, and since $h' \in G_{1,2}$, it must be $\delta$-hyperedge. Also note the negation in the definition of $\mathcal D$: as more paths are introduced, the set $\mathcal D(G)$ gets smaller.

\begin{lemma}\label{lem:fs-D}
  The following is a weak reduction ordering for $\{\FS 1, \FS 2\}$ and a reduction ordering for $\{\FS 3, \FS 4 \}$
  \[
    H \preceq_{\mathcal D} G \iff |\mathcal D(H)| \leq |\mathcal D(G)|
  \]
\end{lemma}

\begin{proof}
  $\mathcal D$ sends a hypergraph to a finite set of ordered pairs, so $\preceq_{\mathcal D}$ is well-founded.

  Note that all four of the rules $\FS i$ preserve the number of $\mu$- and $\delta$-hyperedges in $G$. Hence, if $G \rigidDPOstep{\FS{i}} H$, we can take $H$ to have the same set of hyperedges as $G$, but with different connectivity. For the remainder of the proof, we examine how each rule application affects the paths in a hypergraph. For this, we rely on the fact that all of the hypergraphs involved in the rewriting are monogamous and acyclic. As a consequence of monogamy, any path entering the left-hand side of a rule must do so via an input, and any path exiting must do so via an output.

  For $G \rigidDPOstep{\FS{1}} H$, a $\mu$-hyperedge $h$ has a path to a $\delta$-hyperedge in $G$ if and only if it does in $H$. This follows from the fact that there is a path from every input to the output and from both $\mu$-hyperedges to the output on both sides of the rewriting rule
  \ctikzfig{termination-path-pres}
  Hence $G \sim_{\mathcal D} H$. A symmetric argument holds for $\FS{2}$, so we conclude that $\preceq_{\mathcal D}$ gives a weak reduction ordering for $\{\FS 1, \FS 2\}$.

  For $G \rigidDPOstep{\FS 3} H$, let $L$ be the image of the left-hand side of $\FS{3}$ in $G$ and $R$ the image of the right-hand side of $\FS{3}$ in $H$.
  We will refer to the unique $\mu$-hyperedge in $L$ and $R$ as $h$, and the unique $\delta$-hyperedge in $L$ and $R$ as $h'$.
  First, note that there is a path from every input in $R$ to every output. There is also a path from every input of $R$ to $h'$ and from $h$ to every output
  \ctikzfig{termination-path-reduce}
  Hence, applying $\FS{3}$ can only create more paths from $\mu$-hyperedges to $\delta$-hyperedges and never breaks them, so $\mathcal D(H) \subseteq \mathcal D(G)$. 
  Furthermore, by acyclicity there must not be a path from $h$ to $h'$ in $G$, but there is one in $H$. So the containment $\mathcal D(H) \subseteq \mathcal D(G)$ is strict 
  and thus
  $H \prec_{\mathcal D} G$. The argument for $\FS{4}$ is identical.
\end{proof}

\begin{therm}
  $\FS{}$ is terminating.
\end{therm}

\begin{proof}
  We form the lexicographic ordering $\precFS := \preceq_{\mathcal D, \mathcal L}$. It then follows from Lemmas~\ref{lem:combine-red-ord}, \ref{lem:fs-L}, and \ref{lem:fs-D} that $\precFS$ gives a reduction ordering for $\FS{}$.
\end{proof}

It is worth noting that acylicity plays a crucial role in the above proof. If the two hyperedges in the left-hand side of $\FS{3}$ or $\FS{4}$ were part of a directed cycle, one could potentially find an infinite sequence of rule applications.

Next example shows that, while \textit{convex} rewriting prevents us from introducing cycles (and hence non-terminating behaviour), it also breaks confluence from this system. Our counter-example is based on Example 3.11 from~\cite{power1991npasting}, which was given in terms of string diagrams by~\cite{AmarTweet}. While Power's original example concerned morphisms in a 3-category, the same phenomenon appears in symmetric monoidal categories, and can be understood as a surprising consequence of the convexity condition: namely, even non-overlapping rule applications can block one-another.

\begin{example}
\label{convex-blocking}
Consider the following diagram, and its rendering as a cospan of hypergraphs
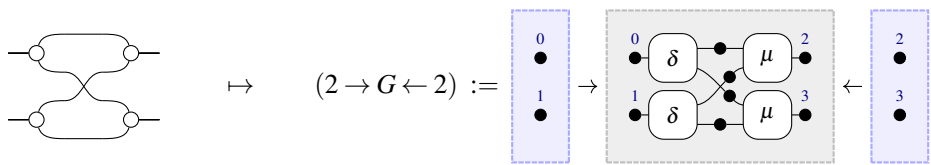
\[ \input{counter-ex.tikz} \qquad \mapsto \qquad (2 \rightarrow G \leftarrow 2) \ :=\ \input{counter-ex-hyp.tikz} \]
For a bialgebra (cf. the next section), this is a familiar diagram, as it is the right-hand side of one of the equations. In that context, it is also a normal form. That is, none of the bialgebra rules can be applied, so this diagram is considered fully simplified. However, for Frobenius semi-algebras there are two different rules that apply: $\FS 3$ and $\FS 4$. Let us have a look at the hypergraph we obtain when we apply each of these two rules
\begin{align}
  G & \ \rigidDPOstep{\FS{3}} \ \input{counter-ex-fs3.tikz}
  \label{eq:counter-ex-fs3}\\
  G & \ \rigidDPOstep{\FS{4}} \ \input{counter-ex-fs4.tikz}
  \label{eq:counter-ex-fs4}
\end{align}
Note that these two rule applications act on disjoint sets of hyperedges, and yet they still interfere with each other. In particular, applying $\FS{3}$ introduces a new path 
from the leftmost $\delta$-hyperedge in hypergraph~\eqref{eq:counter-ex-fs3} to the rightmost $\mu$-hyperedge. Whereas these two hyperedges previously defined a 
convex sub-hypergraph of $G$, they are no longer convex after $\FS{3}$ has been applied. Consequently, these no longer define a valid match for $\FS{4}$. 
Similarly, applying $\FS{4}$ to $G$ in \eqref{eq:counter-ex-fs4} blocks the application of $\FS{3}$. In fact, neither of the hypergraphs \eqref{eq:counter-ex-fs3} or 
\eqref{eq:counter-ex-fs4} contain a match for any of the rules of $\FS{}$. Hence, we have distinct hypergraphs $H_1$ and $H_2$ arising from $G$ that cannot be 
rewritten into the same hypergraph by the rules of $\FS{}$, so $\FS{}$ is not confluent.
\end{example}

\subsection{Bialgebras}
\label{case:bialg}

We now consider \emph{bialgebras}, i.e., a theory with the same generators as a Frobenius algebra
\[
  \left\{\ \mu := \Wmult, \eta := \Wunit, \delta := \Wcomult,\epsilon := \Wcounit\right\}
\]
but with a different set of equations. It is the theory underlying the (bi)-category of spans of sets with disjoint union as monoidal product~\cite{BruniG01},
and it has been used in the axiomatisation of flownomials, an algebraic presentation of flowcharts~\cite{stefanescu2000network}.
The equations of non-commutative bialgebras are given in Figure~\ref{fig:nb-rules}, and their associated DPO rewriting rules, forming the rewriting system $\BA{}$, are shown in Figure~\ref{fig:nb-rules-dpo}.

\begin{figure*}
\begin{equation} \label{eq:ncbialgebralaws}
  \scalebox{1.0}{\input{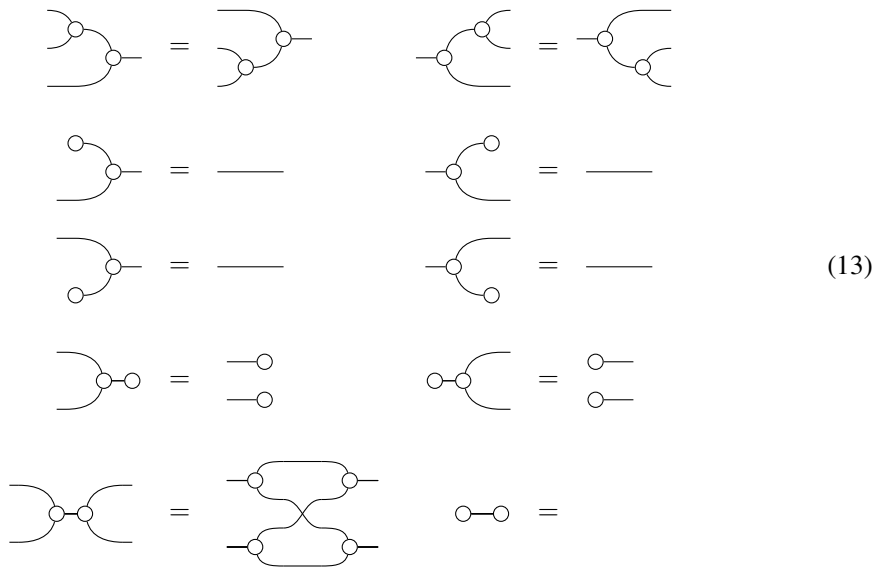}}
\end{equation}

\medskip

\caption{The equations of a bialgebra.}
\label{fig:nb-rules}
\end{figure*}

\begin{figure*}
  \[
    \scalebox{0.7}{\input{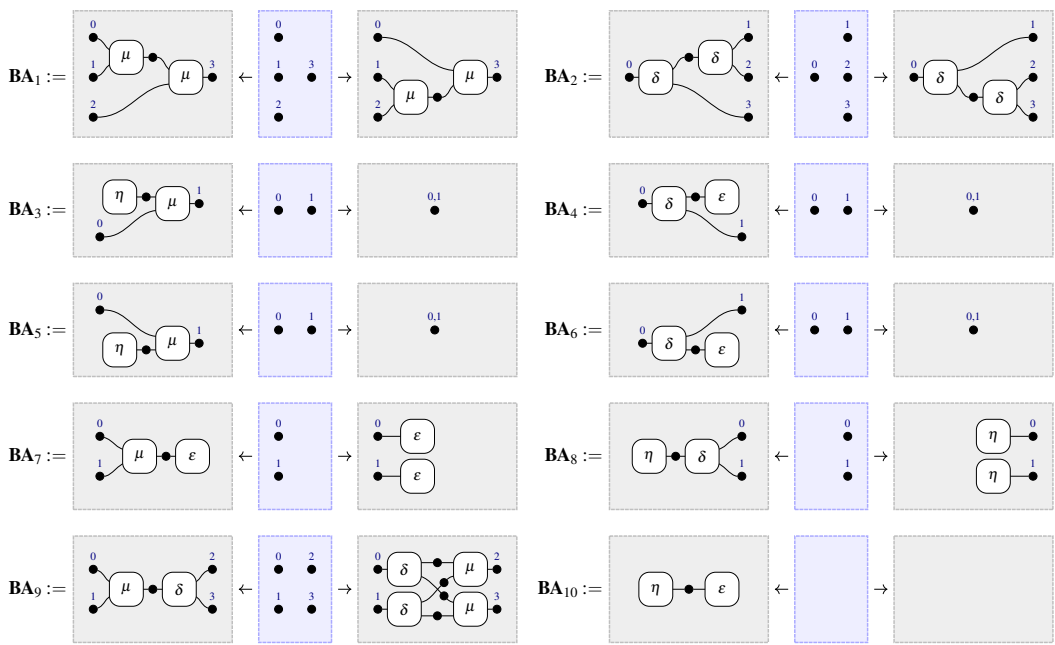}}
  \]

  \medskip

  \caption{DPO rewriting system $\BA{}$ for bialgebras.} \label{fig:nb-rules-dpo}
\end{figure*}

We now focus on proving termination for this system. Its proof is slightly more elaborate, as there are more rules, and the rules do not always preserve the number of hyperedges in a hypergraph. Notably, repeated applications of $\BA{9}$ can significantly increase the number of hyperedges.

Nevertheless, we can find useful reduction orderings by counting paths rather than counting hyperedges. For this, we define two kinds of paths, one that tracks paths involving (co)\underline{U}nits, and the other for (co)\underline{M}ultiplications
\begin{itemize}
  \item a \textit{U-path} is a path $p$ from an input or an $\eta$-hyperedge to an output or an $\epsilon$-hyperedge;
  \item an \textit{M-path} is a path from a $\mu$-hyperedge to a $\delta$-hyperedge.
\end{itemize}

Next, define orders \precU, \precM, \precmu, \precdelta based on counting the number of U-paths, M-paths, $\mu$-hyperedges, and $\delta$-hyperedges, respectively. Using these four orderings, along with $\precL$ defined in the previous section, we define the following lexicographic ordering
\begin{equation}\label{eq:lex-metric}
  \precBA \ :=\  \preceq_{U, M, \mu, \delta, \mathcal L}
\end{equation}
Each of the components of $\precBA$ is well-founded, so $\precBA$ itself is well-founded. Thus we can conclude as follows.

\begin{therm}
  $\precBA$ is a reduction ordering for $\BA{}$, thus $\BA{}$ terminates.
\end{therm}

\begin{proof}
  We argue rule-by-rule, showing that each one is strictly decreasing in one of the orders from \eqref{eq:lex-metric}, and non-increasing in every order that is prior in the lexicographic ordering.

  Since every rule \BA{j} has a unique path from every input to every output for both left- and right-hand side,
  applications of these rules have no effect on paths which start and finish outside of their image. Hence, for each rule,
  we only need to consider paths which start or terminate in the image of the left-hand side.

  \BA{1} has no effect on $\eta$ or $\epsilon$ hyperedges, hence on \precU. No M-path can terminate in \BA{1} and any M-path originating on \BA{1} must exit through the unique output. Since there are precisely two $\mu$-hyperedges in both the left- and the right-hand side, there is a one-to-one correspondence between M-paths before and after applying the rule. \BA{1} leaves the number of $\mu$ and $\delta$ hyperedges fixed, so it suffices to show it strictly decreases \precL.  Applying the rule has no effect on the $\mathcal{L}$-weight of any $\mu$-hyperedges outside of the image of the left-hand side. Suppose there are $\mu$-trees of size $a, b, c$ connected to inputs $0, 1, 2$ of the left-hand side, respectively. The $\mathcal{L}$-weight of the two $\mu$-hyperedges on the left-hand side are thus $a$ and $a + b + 1$, whereas on the right-hand side they are $a$ and $b$. Hence \precL is strictly decreased.
  \BA{2} follows via a symmetric argument.

  Since \BA{3}--\BA{6} and \BA{10} remove $\eta$- and $\epsilon$-hyperedges from the hypergraph, they will strictly decrease the number of U-paths.

  For \BA{7}, no U-path can terminate in the left-hand side, and any U-path starting in the left-hand side must exit through one of the two outputs. Hence it corresponds to a unique U-path exiting the right-hand side. M-paths are unaffected, as is the number of $\delta$-hyperedges. However, the number of $\mu$-hyperedges is strictly decreased, so \BA{7} strictly decreases \precmu. The argument for \BA{8} is symmetric, yet with respect to \precdelta.

  \BA{9} has no $\eta$ or $\epsilon$-hyperedges in either the left- or the right-hand side, so it leaves the number of U-paths fixed.  Consider an M-path that enters the left-hand side from the left. It enters either from input $0$ or input $1$, hence it corresponds to a unique M-path entering the right-hand side. We can argue similarly for M-paths exiting on the right. Hence, the only M-path left to consider is the one from the $\mu$-hyperedge to the $\delta$-hyperedge in the left-hand side, which is eliminated. Thus \BA{9} strictly reduces \precM.
\end{proof}

It is also possible to show that, unlike $\FS{}$, the rewriting system $\BA{}$ is confluent. An important factor in the confluence proof is the fact that the rewriting 
system $\BA{}$ is left-connected (cf. Definition~\ref{def:left-connected}), so we do not need to impose convexity as an additional requirement when we do rewriting, 
thanks to Theorem~\ref{thm:stronglyconnected}. This rules out situations like the one for $\FS{}$ in Example~\ref{convex-blocking}, where disjoint rule applications 
can block one-another due to convexity considerations.


In order to show confluence, we can use a technique known as \textit{critical pair analysis}.
There are various subtleties arising in critical pair analysis for DPO rewriting~\cite{Plump1993} and general rewriting for symmetric monoidal categories~\cite{Lafont2003}, which are beyond the scope of the this paper. Hence, we leave a formal proof of the confluence of $\BA{}$ for the sequel to this paper, in which we develop a comprehensive  framework for critical pair analysis on convex DPOI rewriting~\cite{BGKSZ-partthree}.

\section{Conclusions and Further Works}
\label{sec:conclusions}

In this paper we developed a practical approach to the rewriting of symmetric monoidal categories. Relying on a previously identified~\cite{BGKSZ-partone} correspondence between string diagrams and cospans of hypergraphs, we classify those cospans that do not rely on the presence of an additional Frobenius structure, i.e.\ those that are relevant when considering only symmetric monoidal categories.
Having thereby identified a combinatorial structure that serves as a sound encoding of string diagrams, we use the mechanism of double pushout rewriting, which we modify in order to ensure soundness and completeness. This involves the identification of sufficiently well-behaved pushout complements, and restricting to similarly well-behaved matches: roughly speaking, these restrictions ensure that the rewrites themselves can only rely on the ``vanilla'' symmetric monoidal structure, without any use of the laws of Frobenius algebras. We arrive at a practical procedure for rewriting modulo symmetric monoidal laws: assuming an implementation of DPO rewriting of hypergraphs, each restriction can be easily checked algorithmically.

While originating in category theory~\cite{Joyal1991}, string diagrams have been influential in computer science, especially after the paper on traced monoidal categories by Joyal, Street and Verity~\cite{Joyal_tracedcategories}. However, the correspondence between terms of ``2-dimensional'' algebraic structures---i.e.\ those with sequential and parallel composition, understood as arrows of a free symmetric monoidal category---and suitable hypergraphs (flow diagrams) was recognized earlier, and studied at least since the work of Stefanescu (see the references in the survey~\cite{Selinger2009}).

Closely related to our work, Dixon and Kissinger~\cite{DixonK13} use cospans of string graphs (called there open graphs) to encode morphisms in a symmetric monoidal category
and reason equationally via DPO rewriting. There is an evident encoding of the hypergraphs we use into string graphs.
However, the notion of rewriting considered there is only sound if there is a trace on the symmetric monoidal category, whereas our notion of convex DPO rewriting
guarantees soundness for any symmetric monoidal category. Another difference is that we directly work in an adhesive category, while the category of open graphs
inherits the relevant properties from an embedding into the adhesive category of typed graphs.

In logic and computer science, diagrammatic rewriting was
motivated in part by computational patterns appaearing in the proof theory of linear logic, leading to general diagrammatic rewriting frameworks such as interaction nets~\cite{lafont1989interaction, mazza2006interaction}.
A general rewriting theory of such structures has also been developed; notably
Burroni~\cite{Burroni1993} generalised term rewriting to higher dimensions, including the 3-dimensional case of string diagram rewriting; see Mimram's survey~\cite{Mimram14}.
Here, in order to capture symmetric monoidal structure, the laws of symmetric monoidal categories would usually be considered as explicit rewriting rules, resulting in sophisticated rewriting systems whose analysis is often challenging (see e.g.~\cite{guiraud2006termination,Lafont2003}). Abstract higher-dimensional rewriting is far more general than our approach, which has been tailored over symmetric monoidal categories. The loss of generality brings the benefits of specialisation: our approach has the laws of symmetric monoidal categories built-in, reducing the complexity of the resulting rewriting systems. Thus, our work can be seen as part of a more general effort the search for characterisations that capture some parts of relevant algebraic structure in combinatorial models that bring the possibility of implementation; see e.g.\ Obradovic's work on capturing the algebra of cyclic operads~\cite{Obradovic:phdthesis,DBLP:journals/acs/CurienO20} and Hadzihasanovic's recent work on diagrammatic sets~\cite{DBLP:journals/acs/Hadzihasanovic20,DBLP:conf/lics/Hadzihasanovic21}.

This is the second of a three paper series, the first being~\cite{BGKSZ-partone}, and the third~\cite{BGKSZ-partthree} devoted to solving the problem of confluence for string diagram rewriting. With these papers we hope to lay the foundations for the next generation of diagrammatic proof assistants for symmetric monoidal theories. Such tools would lie between Globular~\cite{DBLP:journals/lmcs/BarKV18} and \texttt{homotopy.io} (\url{https://homotopy.io/}) on the one hand, the foundations of which are designed for reasoning about higher-dimensional weak structures and thus have minimal algebraic structure built-in, and Quantomatic~\cite{DBLP:journals/corr/KissingerZ15a}, PyZX and QuiZX on the other, in which the implementations of rewriting rely on the rich algebraic structure of the ZX-calculus.

%

\bibliographystyle{plain}
\bibliography{catBib3Rev}

\begin{thebibliography}{10}

\bibitem{abramsky2012h}
Samson Abramsky and Chris Heunen.
\newblock H*-algebras and nonunital {F}robenius algebras: First steps in
  infinite-dimensional categorical quantum mechanics.
\newblock In Samson Abramsky and Mike Mislove, editors, {\em Mathematical
  Foundations of Information Flow}, volume~71 of {\em Proceedings of Symposia
  in Applied Mathematics}, pages 1--24, 2012.

\bibitem{BaezErbele-CategoriesInControl}
John Baez and Jason Erbele.
\newblock Categories in control.
\newblock {\em Theory and Application of Categories}, 30:836--881, 2015.

\bibitem{DBLP:journals/lmcs/BarKV18}
Krzysztof Bar, Aleks Kissinger, and Jamie Vicary.
\newblock Globular: an online proof assistant for higher-dimensional rewriting.
\newblock {\em Logical Methods in Computer Science}, 14(1), 2018.

\bibitem{DBLP:journals/entcs/BehrischKP12}
Mike Behrisch, Sebastian Kerkhoff, and John Power.
\newblock Category theoretic understandings of universal algebra and its dual:
  Monads and {L}awvere theories, comonads and what?
\newblock In Ulrich Berger and Michael~W. Mislove, editors, {\em MFPS 2012},
  volume 286 of {\em ENTCS}, pages 5--16. Elsevier, 2012.

\bibitem{BonchiGKSZ_lics16}
Filippo Bonchi, Fabio Gadducci, Aleks Kissinger, Pawel Soboci\'{n}ski, and
  Fabio Zanasi.
\newblock Rewriting modulo symmetric monoidal structure.
\newblock In Martin Grohe, Eric Koskinen, and Natarajan Shankar, editors, {\em
  LICS 2016}, pages 710--719. ACM, 2016.

\bibitem{BonchiGKSZ18}
Filippo Bonchi, Fabio Gadducci, Aleks Kissinger, Pawel Soboci\'{n}ski, and
  Fabio Zanasi.
\newblock Rewriting with {F}robenius.
\newblock In Anuj Dawar and Erich Gr{\"{a}}del, editors, {\em LICS 2018}, pages
  165--174. {ACM}, 2018.

\bibitem{BGKSZ-partone}
Filippo Bonchi, Fabio Gadducci, Aleks Kissinger, Pawel Soboci\'{n}ski, and
  Fabio Zanasi.
\newblock String diagram rewrite theory {I}: Rewriting with {F}robenius
  structure.
\newblock {\em Journal of the ACM}, 69(2):14:1--14:58, 2022.

\bibitem{BGKSZ-partthree}
Filippo Bonchi, Fabio Gadducci, Aleks Kissinger, Pawel Soboci\'{n}ski, and
  Fabio Zanasi.
\newblock String diagram rewrite theory {III}: Confluence with and without
  {F}robenius.
\newblock {\em Mathematical Structures in Computer Science}, 2022.
\newblock to appear.

\bibitem{BGKSZ-esop17}
Filippo Bonchi, Fabio Gadducci, Aleks Kissinger, Pawel~Pawel Soboci\'{n}ski,
  and Fabio Zanasi.
\newblock Confluence of graph rewriting with interfaces.
\newblock In Hongseok Yang, editor, {\em ESOP 2017}, volume 10201 of {\em
  LNCS}, pages 141--169. Springer, 2017.

\bibitem{BonchiGK09}
Filippo Bonchi, Fabio Gadducci, and Barbara K{\"{o}}nig.
\newblock Synthesising {CCS} bisimulation using graph rewriting.
\newblock {\em Information and Computation}, 207(1):14--40, 2009.

\bibitem{BonchiHPS17}
Filippo Bonchi, Joshua Holland, Dusko Pavlovic, and Pawel Soboci\'{n}ski.
\newblock Refinement for signal flow graphs.
\newblock In Roland Meyer and Uwe Nestmann, editors, {\em CONCUR 2017},
  volume~85 of {\em LIPIcs}, pages 24:1--24:16. Schloss Dagstuhl -
  Leibniz-Zentrum f{\"{u}}r Informatik, 2017.

\bibitem{Bonchi2017c}
Filippo Bonchi, Dusko Pavlovic, and Pawel Soboci\'{n}ski.
\newblock Functorial semantics for relational theories.
\newblock {\em Preprint available at arXiv:1711.08699}, 2017.

\bibitem{BPSZ-lics19}
Filippo Bonchi, Robin Piedeleu, Pawel Soboci\'{n}ski, and Fabio Zanasi.
\newblock Graphical affine algebra.
\newblock In {\em LICS 2019}, pages 1--12. IEEE, 2019.

\bibitem{BonchiSZ17}
Filippo Bonchi, Pawel Soboci\'{n}ski, and Fabio Zanasi.
\newblock The calculus of signal flow diagrams {I:} {L}inear relations on
  streams.
\newblock {\em Information and Computation}, 252:2--29, 2017.

\bibitem{Bonchi2018}
Filippo Bonchi, Pawel Soboci\'{n}ski, and Fabio Zanasi.
\newblock Deconstructing {L}awvere with distributive laws.
\newblock {\em Logic and Algebraic Methods in Programming}, 95:128--146, 2018.

\bibitem{BruniG01}
Roberto Bruni and Fabio Gadducci.
\newblock Some algebraic laws for spans (and their connections with
  multirelations).
\newblock In Wolfram Kahl, David~Lorge Parnas, and Gunther Schmidt, editors,
  {\em RELMIS 2001}, volume 44(3) of {\em ENTCS}, pages 175--193. Elsevier,
  2001.

\bibitem{Burroni1993}
Albert Burroni.
\newblock Higher dimensional word problems with applications to equational
  logic.
\newblock {\em Theoretical Computer Science}, 115(1):43--62, 1993.

\bibitem{CoeckeDuncanZX2011}
Bob Coecke and Ross Duncan.
\newblock Interacting quantum observables: Categorical algebra and
  diagrammatics.
\newblock {\em New Journal of Physics}, 13(4):1--85, 2011.

\bibitem{CoeckeKissinger_book}
Bob Coecke and Aleks Kissinger.
\newblock {\em Picturing Quantum Processes: {A} First Course in Quantum Theory
  and Diagrammatic Reasoning}.
\newblock Cambridge University Press, 2017.

\bibitem{GadducciCorradini-functorial}
Andrea Corradini and Fabio Gadducci.
\newblock Functorial semantics for multi-algebras and partial algebras, with
  applications to syntax.
\newblock {\em Theoretical Computer Science}, 286(2):293--322, 2002.

\bibitem{HandbookDPO}
Andrea Corradini, Ugo Montanari, Francesca Rossi, Hartmut Ehrig, Reiko Heckel,
  and Michael L{\"{o}}we.
\newblock Algebraic approaches to graph transformation - {P}art {I}: Basic
  concepts and double pushout approach.
\newblock In Grzegorz Rozenberg, editor, {\em Handbook of Graph Grammars and
  Computing by Graph Transformations, Volume 1: Foundations}, pages 163--246.
  World Scientific, 1997.

\bibitem{DBLP:journals/acs/CurienO20}
Pierre{-}Louis Curien and Jovana Obradovic.
\newblock Categorified cyclic operads.
\newblock {\em Applied Categorical Structures}, 28(1):59--112, 2020.

\bibitem{DixonK13}
Lucas Dixon and Aleks Kissinger.
\newblock Open-graphs and monoidal theories.
\newblock {\em Mathematical Structures in Computer Science}, 23(2):308--359,
  2013.

\bibitem{Ehrig2004}
Hartmut Ehrig and Barbara K{\"{o}}nig.
\newblock Deriving bisimulation congruences in the {DPO} approach to graph
  rewriting.
\newblock In Igor Walukiewicz, editor, {\em FOSSACS 2004}, volume 2987 of {\em
  LNCS}, pages 151--166. Springer, 2004.

\bibitem{Fox76}
Thomas Fox.
\newblock Coalgebras and cartesian categories.
\newblock {\em Communications in Algebra}, 4:665--667, 1976.

\bibitem{Gadducci07}
Fabio Gadducci.
\newblock Graph rewriting for the $\pi$-calculus.
\newblock {\em Mathematical Structures in Computer Science}, 17(3):407--437,
  2007.

\bibitem{Gadducci1998}
Fabio Gadducci and Reiko Heckel.
\newblock An inductive view of graph transformation.
\newblock In Francesco Parisi{-}Presicce, editor, {\em WADT 1997}, volume 1376,
  pages 223--237. Springer, 1998.

\bibitem{DBLP:journals/lmcs/GarnerP18}
Richard Garner and John Power.
\newblock An enriched view on the extended finitary monad-{L}awvere theory
  correspondence.
\newblock {\em Logical Methods in Computer Science}, 14(1), 2018.

\bibitem{GhicaJL17}
Dan~R. Ghica, Achim Jung, and Aliaume Lopez.
\newblock Diagrammatic semantics for digital circuits.
\newblock In Valentin Goranko and Mads Dam, editors, {\em CSL 2017}, volume~82
  of {\em LIPIcs}, pages 24:1--24:16. Schloss Dagstuhl - Leibniz-Zentrum
  f{\"{u}}r Informatik, 2017.

\bibitem{guiraud2006termination}
Yves Guiraud.
\newblock Termination orders for three-dimensional rewriting.
\newblock {\em Pure and Applied Algebra}, 207(2):341--371, 2006.

\bibitem{AmarTweet}
Amar Hadzihasanovic.
\newblock Via Twitter:
  \href{https://twitter.com/amar_hh/status/1336274654923788288?s=20}{\color{blue}
  twitter.com/amar\_hh/status/1336274654923788288?s=20}, 2020.

\bibitem{DBLP:journals/acs/Hadzihasanovic20}
Amar Hadzihasanovic.
\newblock A combinatorial-topological shape category for polygraphs.
\newblock {\em Applied Categorical Structures}, 28(3):419--476, 2020.

\bibitem{DBLP:conf/lics/Hadzihasanovic21}
Amar Hadzihasanovic.
\newblock The smash product of monoidal theories.
\newblock In {\em LICS 2021}, pages 1--13. {IEEE}, 2021.

\bibitem{DBLP:conf/ifipTCS/HylandPP02}
Martin Hyland, Gordon~D. Plotkin, and John Power.
\newblock Combining computational effects: {C}ommutativity {\&} sum.
\newblock In Ricardo~A. Baeza{-}Yates, Ugo Montanari, and Nicola Santoro,
  editors, {\em TCS 2002}, volume 223 of {\em {IFIP} Conference Proceedings},
  pages 474--484. Kluwer, 2002.

\bibitem{hyland2007category}
Martin Hyland and John Power.
\newblock The category theoretic understanding of universal algebra: Lawvere
  theories and monads.
\newblock In Luca Cardelli, Marcelo~P. Fiore, and Glynn Winskel, editors, {\em
  Computation, Meaning, and Logic}, volume 172 of {\em ENTCS}, pages 437--458.
  Elsevier, 2007.

\bibitem{JacobsKZ19}
Bart Jacobs, Aleks Kissinger, and Fabio Zanasi.
\newblock Causal inference by string diagram surgery.
\newblock In Mikolaj Bojanczyk and Alex Simpson, editors, {\em FOSSACS 2019},
  volume 11425 of {\em LNCS}, pages 313--329. Springer, 2019.

\bibitem{Joyal1991}
Andre Joyal and Ross Street.
\newblock The geometry of tensor calculus, \textsc{i}.
\newblock {\em Advances in Mathematics}, 88(1):55--112, 1991.

\bibitem{Joyal_tracedcategories}
Andr{\'e} Joyal, Ross Street, and Dominic Verity.
\newblock Traced monoidal categories.
\newblock {\em Mathematical Proceedings of the Cambridge Philosophical
  Society}, 119(3):447--468, 1996.

\bibitem{DBLP:journals/corr/KissingerZ15a}
Aleks Kissinger and Vladimir Zamdzhiev.
\newblock Quantomatic: {A} proof assistant for diagrammatic reasoning.
\newblock {\em Preprint available at arXiv:1503.01034}, 2015.

\bibitem{DBLP:journals/jfp/LackP09}
Stephen Lack and John Power.
\newblock Gabriel-{U}lmer duality and {L}awvere theories enriched over a
  general base.
\newblock {\em Functional Programming}, 19(3-4):265--286, 2009.

\bibitem{Lack2004a}
Steve Lack.
\newblock Composing {PROPs}.
\newblock {\em Theory and Application of Categories}, 13(9):147--163, 2004.

\bibitem{Lack2005}
Steve Lack and Pawel Soboci\'{n}ski.
\newblock Adhesive and quasiadhesive categories.
\newblock {\em Theoretical Informatics and Applications}, 39(3):511--546, 2005.

\bibitem{lafont1989interaction}
Yves Lafont.
\newblock Interaction nets.
\newblock In Frances~E. Allen, editor, {\em POPL 1990}, pages 95--108, 1990.

\bibitem{Lafont2003}
Yves Lafont.
\newblock Towards an algebraic theory of {B}oolean circuits.
\newblock {\em Pure and Applied Algebra}, 184(2--3):257--310, 2003.

\bibitem{Liberti2021}
Ivan~Di Liberti, Fosco Loregian, Chad Nester, and Pawel Soboci\'{n}ski.
\newblock Functorial semantics for partial theories.
\newblock In {\em POPL 2021}, 2021.

\bibitem{MacLane1965}
Saunders {Mac Lane}.
\newblock Categorical algebra.
\newblock {\em Bulletin of the American Mathematical Society}, 71(1):40--106,
  1965.

\bibitem{mazza2006interaction}
Damiano Mazza.
\newblock {\em Interaction nets: Semantics and concurrent extensions}.
\newblock PhD thesis, Universit{\'e} Aix-Marseille II/Universit\`a degli Studi
  Roma Tre, 2006.

\bibitem{Mimram14}
Samuel Mimram.
\newblock Towards 3-dimensional rewriting theory.
\newblock {\em Logical Methods in Computer Science}, 10(2), 2014.

\bibitem{moggi1991notions}
Eugenio Moggi.
\newblock Notions of computation and monads.
\newblock {\em Information and Computation}, 93(1):55--92, 1991.

\bibitem{Obradovic:phdthesis}
Jovana Obradovic.
\newblock {\em Cyclic operads: Syntactic, algebraic and categorified aspects}.
\newblock PhD thesis, Universit\'e Paris 7, 2017.

\bibitem{DBLP:conf/fossacs/PlotkinP01}
Gordon~D. Plotkin and John Power.
\newblock Adequacy for algebraic effects.
\newblock In Furio Honsell and Marino Miculan, editors, {\em FOSSACS 2001},
  volume 2030 of {\em LNCS}, pages 1--24. Springer, 2001.

\bibitem{plotkin2001semantics}
Gordon~D. Plotkin and John Power.
\newblock Semantics for algebraic operations.
\newblock In Stephen~D. Brookes and Michael~W. Mislove, editors, {\em MFPS
  2001}, volume~45 of {\em ENTCS}, pages 332--345. Elsevier, 2001.

\bibitem{DBLP:conf/fossacs/PlotkinP02}
Gordon~D. Plotkin and John Power.
\newblock Notions of computation determine monads.
\newblock In Mogens Nielsen and Uffe Engberg, editors, {\em FOSSACS 2002},
  volume 2303 of {\em LNCS}, pages 342--356. Springer, 2002.

\bibitem{DBLP:journals/acs/PlotkinP03}
Gordon~D. Plotkin and John Power.
\newblock Algebraic operations and generic effects.
\newblock {\em Applied Categorical Structures}, 11(1):69--94, 2003.

\bibitem{DBLP:journals/entcs/PlotkinP04}
Gordon~D. Plotkin and John Power.
\newblock Computational effects and operations: An overview.
\newblock In Mart\'in~H\"otzel Escard\'o and Achim Jung, editors, {\em Workshop
  on Domains VI}, volume~73 of {\em ENTCS}, pages 149--163. Elsevier, 2004.

\bibitem{Plump1993}
Detlef Plump.
\newblock Hypergraph rewriting: Critical pairs and undecidability of
  confluence.
\newblock In M.~Ronan Sleep, Marinus~J. Plasmeijer, and Marko~C.J.D. van
  Eekele, editors, {\em Term Graph Rewriting: Theory and Practice}, pages
  201--213. Wiley, 1993.

\bibitem{power1991npasting}
John Power.
\newblock An n-categorical pasting theorem.
\newblock In Aurelio Carboni, Maria~Cristina Pedicchio, and Giuseppe Rosolini,
  editors, {\em Category theory}, Lecture Notes in Mathematics, pages 326--358.
  Springer, 1991.

\bibitem{power1999enriched}
John Power.
\newblock Enriched {L}awvere theories.
\newblock {\em Theory and Applications of Categories}, 6(7):83--93, 1999.

\bibitem{DBLP:conf/fossacs/Power04}
John Power.
\newblock Canonical models for computational effects.
\newblock In Igor Walukiewicz, editor, {\em FOSSACS 2004}, volume 2987 of {\em
  LNCS}, pages 438--452. Springer, 2004.

\bibitem{DBLP:conf/calco/Power05}
John Power.
\newblock Discrete {L}awvere theories.
\newblock In Jos{\'{e}}~Luiz Fiadeiro, Neil Harman, Markus Roggenbach, and Jan
  J. M.~M. Rutten, editors, {\em CALCO 2005}, volume 3629 of {\em LNCS}, pages
  348--363. Springer, 2005.

\bibitem{DBLP:journals/entcs/Power06a}
John Power.
\newblock Countable {L}awvere theories and computational effects.
\newblock In Anthony~Karel Seda, Ted Hurley, Michel~P. Schellekens,
  M\'iche\'al~Mac an~Airchinnigh, and Glenn Strong, editors, {\em MFCSIT 2004},
  volume 161 of {\em ENTCS}, pages 59--71. Elsevier, 2006.

\bibitem{DBLP:conf/mpc/Power06}
John Power.
\newblock The universal algebra of computational effects: {L}awvere theories
  and monads.
\newblock In Conor McBride and Tarmo Uustalu, editors, {\em MSFP@MPC 2006},
  Workshops in Computing. {BCS}, 2006.

\bibitem{SadrzadehCC14}
Mehrnoosh Sadrzadeh, Stephen Clark, and Bob Coecke.
\newblock The {F}robenius anatomy of word meanings {I}: Subject and object
  relative pronouns.
\newblock {\em Logic and Computation}, 23(6):1293--1317, 2013.

\bibitem{Sassone2005a}
Vladimiro Sassone and Pawel Soboci\'{n}ski.
\newblock Reactive systems over cospans.
\newblock In {\em LICS 2005}, pages 311--320. {IEEE} Computer Society, 2005.

\bibitem{Selinger2009}
Peter Selinger.
\newblock A survey of graphical languages for monoidal categories.
\newblock {\em Springer Lecture Notes in Physics}, 13(813):289--355, 2011.

\bibitem{stefanescu2000network}
Gheorge Stefanescu.
\newblock {\em Network Algebra}.
\newblock Springer, 2000.

\bibitem{Terese03}
Terese.
\newblock {\em Term Rewriting Systems}.
\newblock Cambridge University Press, 2003.

\bibitem{ullmann1976algorithm}
Julian~R Ullmann.
\newblock An algorithm for subgraph isomorphism.
\newblock {\em Journal of the ACM}, 23(1):31--42, 1976.

\bibitem{ZanasiThesis}
Fabio Zanasi.
\newblock {\em Interacting {H}opf Algebras: The Theory of Linear Systems}.
\newblock PhD thesis, \'Ecole Normale Sup\'{e}rieure de Lyon, 2015.

\bibitem{Zanasi16}
Fabio Zanasi.
\newblock The algebra of partial equivalence relations.
\newblock In Lars Birkedal, editor, {\em MFPS 2016}, volume 325 of {\em ENTCS},
  pages 313--333. Elsevier, 2016.

\end{thebibliography}
\end{document}